\pdfoutput=1
\documentclass[11pt]{scrartcl}

\usepackage[a4paper, total={16cm, 24cm}]{geometry}
\usepackage{microtype}
\usepackage{multicol}

\usepackage{xspace}

\usepackage{amsmath}
\usepackage{amssymb}
\usepackage{amsthm}
\usepackage{nicefrac}
\usepackage{thm-restate}
\usepackage{mathtools}

\renewcommand\tableofcontents{\listoftoc*{toc}} %

\usepackage[dvipsnames]{xcolor}
\usepackage{graphicx}
\usepackage{tikz}
\usetikzlibrary{backgrounds,calc}

\usepackage{natbib}
\usepackage{hyperref}
\hypersetup{
    pdfencoding=auto, 
    psdextra,
    colorlinks=true,
    citecolor=green!40!black,
    linkcolor=red!50!black,
    urlcolor=blue!80!black
}
\usepackage[nameinlink]{cleveref}
\usepackage{doi}

\usepackage{tabularx}
\usepackage{booktabs}
\usepackage{multirow}

\usepackage{algorithm}
\usepackage{algorithmic}

\usepackage{pdfpages}
\usepackage{multicol}

\newcommand{\N}{\mathbb{N}}

\newcommand{\Yes}{\textsc{Yes}\xspace}
\newcommand{\No}{\textsc{No}\xspace}

\newcommand{\agents}{A}
\newcommand{\numAgents}{{|\agents|}}

\newcommand{\alloc}{\pi}

\newcommand{\dFn}{\operatorname{d}}
\newcommand{\relSet}{M}

\DeclareMathOperator{\cost}{cost}

\newcommand{\Gpref}{D_{\operatorname{pref}}}

\newcommand{\probName}[1]{\textsc{#1}\xspace}

\DeclareMathOperator{\dist}{dist}

\newcommand{\Oh}[1]{{\mathcal{O}\left(#1\right)}}

\NewDocumentCommand{\cc}{ O{} O{} m }{\mbox{%
    \expandafter\ifx\expandafter\relax\detokenize{#2}\relax\else{#2-}\fi%
    \textsf{#3}%
    \expandafter\ifx\expandafter\relax\detokenize{#1}\relax\else{-#1}\fi%
    }\xspace}

\newcommand{\NPc}{\cc[complete]{NP}}
\newcommand{\NPcness}{\cc[completeness]{NP}}
\newcommand{\FPT}{\cc{FPT}}
\newcommand{\XP}{\cc{XP}}
\newcommand{\W}[1][1]{\cc{W[#1]}}
\newcommand{\Wh}[1][1]{\cc[hard]{W[#1]}}

\newcommand{\PLS}{\cc{PLS}}

\newtheorem{theorem}{Theorem}
\newtheorem{lemma}{Lemma}
\newtheorem{corollary}{Corollary}
\newtheorem{observation}{Observation}
\newtheorem{proposition}{Proposition}
\newtheorem{definition}{Definition}
\newtheorem{example}{Example}
\newtheorem{claim}{Claim}
\newenvironment{claimproof}[1][Proof.]{\noindent\emph{#1}\hspace{0.15cm}}{\hfill~$\blacktriangleleft$\medskip}
\newtheorem{rrule}{Reduction Rule}

\newtheorem{op}{Open Problem}
\newenvironment{openProblem}{\begin{op}}{\end{op}}

\definecolor{cbBlue}{HTML}{332288}
\definecolor{cbOrange}{HTML}{DDCC77}
\definecolor{cbGreen}{HTML}{44AA99}
\definecolor{cbRed}{HTML}{882255}

\title{Stability in Distance Preservation Games on Graphs\footnote{An extended abstract of this work has been published in the Proceedings of the 25th International Conference on Autonomous Agents and Multiagent Systems, {AAMAS}~'26~\citep{DeligkasEGKS2026}.}}

\usepackage{authblk}

\author[1]{Argyrios Deligkas}
\author[1]{Eduard Eiben}
\author[1]{Tiger-Lily Goldsmith}
\author[2]{\authorcr{}Dušan Knop}
\author[2,3]{Šimon Schierreich}
\affil[1]{Royal Holloway, University of London}
\affil[2]{Czech Technical University in Prague} 
\affil[3]{AGH University of Krakow}

\begin{document}

\maketitle

\begin{abstract}
	\begin{center}
		\textbf{\textsf{Abstract}} \smallskip
	\end{center}
    We introduce a new class of network allocation games called \emph{graphical distance preservation games}. Here, we are given a graph, called a topology, and a set of agents that need to be allocated to its vertices. Moreover, every agent has an ideal (and possibly different) distance in which to be from some of the other agents. Given an allocation of agents, each one of them suffers a cost that is the sum of the differences from the ideal distance for each agent in their subset. The goal is to decide whether there is a stable allocation of the agents, i.e., no agent would like to deviate from their location. Specifically, we consider three different stability notions: envy-freeness, swap stability, and jump stability. We perform a comprehensive study of the (parameterized) complexity of the problem in three different dimensions: the topology of the graph, the number of agents, and the structure of preferences of the agents.
\end{abstract}

\vspace{4pt}
\hrule
\vspace{4pt}
{\small\tableofcontents}
\vspace{8pt}
\hrule
\newpage

\section{Introduction}
It is the time of year to organize the annual banquet of your organization, and your task is to allocate the seats for the attendees. 
Of course, you could choose the seats in an arbitrary way, but you know that this is not the smartest idea. 
You are aware that a) particular people should ideally be seated at specific distances -- not too close, but not too far either, so they can keep an eye on each other -- and b) the allocation should be ``{\em stable}'', e.g., chatty-Bob cannot find an empty seat
close to grumpy-Joe, or complaining-Jack will not envy cool-Bill for his seat, or adventurous-Alice and poetical-Ada will not both benefit from swapping their seats, as they like their current seats more.

Situations like the above occur in many scenarios, ranging from office and house allocations to the positioning of supervisors and employees on assembly lines. 
In each one of these scenarios, there is an underlying network whose nodes correspond to available positions, and a set of agents, each with their own preferred distance for the subset of agents they are interested in. 
For example, a supervisor wants to be able to reach each of their supervisees within at most three minutes, while each supervisee wants to be next to their friends, but at a distance of fifty meters from their supervisor. 
The goal, as argued in our initial example, is to place the agents on the nodes of the graph such that the allocation is stable.

Motivated by similar problems, very recently, \citet{AzizCLNW2025} introduced an elegant model, called {\em distance preservation games}, in order to formally study them.
In their model, there was a set of agents, each of whom had an ideal distance for a subset of the remaining agents. 
The goal was to place the agents on the real unit-interval such that an objective associated with the costs of the agents was met; the cost of an agent is the sum of the differences from their ideal distance for every agent in their subset.
They have studied social cost, i.e., the sum of agents' costs, and {\em jump stable} placements, i.e., no agent would like to jump to a different position of the interval.
Nevertheless, there exist many scenarios, like some of the ones above, that cannot be captured by this model, due to its domain: we can have a more complex space than the interval, or it may not be physically possible to place an agent anywhere we want.

\subsection{Our Contribution}
Our contribution is threefold. 
First, we introduce {\em graphical distance preservation games} by augmenting the model of \citet{AzizCLNW2025}.
Now, instead of the line segment, we have a {\em topology}, i.e., a graph, where we have to allocate the agents on its vertices. 
Then, we perform a comprehensive study of the (parameterized) complexity of finding envy-free allocations on this type of game.
Finally, we initiate the study of jump stable and swap stable allocations. 

Under envy-free allocations, 
we first observe that while they are guaranteed to exist when there are only two agents (\Cref{prop:EF:two_agents}), this is no longer true for three agents or more (\Cref{thm:EF:notGuaranteed}).
In fact, 
we prove that the problem is \NPc even when the agents have symmetric preferences, i.e., for every pair of agents~$i, j$, the ideal distance of~$i$ from~$j$ matches the ideal distance of~$j$ from~$i$, and furthermore every ideal distance is equal to 1 (\Cref{thm:EF:NPh}). Hence, in order to hope for tractability we have to impose some constraints on the instances. We do this in three different dimensions.
Additionally, we observe that there exist games that admit swap stable solutions, but no envy-free ones~(\Cref{thm:swapAndNotEF}).

The first restriction we consider is on the structure of the topology,~$G$.
We begin by proving that the problem can be solved in polynomial time when the underlying graph is a clique or a star (\Cref{thm:EF:star_clique}).
Unfortunately, this is possibly the best we can hope for under this restriction -- at least for ``standard'' structural parameters -- as we prove that the problem becomes \NPc if~$G$ has vertex cover number 2 (\Cref{thm:EF:pNPh:vc}), or if~$G$ is a tree of depth 2 and agents have symmetric preferences (\Cref{thm:NPh:tree}), or if it is a path (\Cref{thm:EF:NPh:path}).
These results strongly indicate that we need a different restriction in order to get tractability.

Our second dimension of restriction is on the number of agents. 
First, we prove that the problem remains \NPc even when the size of the graph is significantly larger than the number of agents (\Cref{thm:EF:NPh:boundedAgents}) and we complement this negative result with an \XP algorithm parameterized by the number of agents (\Cref{thm:EF:XP:numAgents}). Hence, the problem is efficiently-solvable for any constant number of agents.
Whether this \XP algorithm can be improved to \FPT or there is a complementary \W[1]-hardness remains open.
On the positive side though, if we consider as parameters the number of agents and structural parameters of the topology,~$G$, we can design a number of fixed parameter algorithms. 
These structural parameters include: 
the vertex cover number of~$G$ (\Cref{thm:EF:FPT:agentsPlusVC}), which complements the negative result of \Cref{thm:EF:pNPh:vc};
the {\em neighborhood diversity} of~$G$ (\Cref{thm:EF:nd:agents:FPT}); 
the {\em modular-width} of~$G$ (\Cref{thm:EF:mw:agents:FPT}), which to the best of our knowledge is the first time that it is used within computational social choice;
the {\em diameter} of~$G$ (\Cref{thm:EF:diameter:agents:FPT}), which complements \Cref{thm:NPh:tree} that proves hardness for topologies with diameter~$4$;
the distance to clique (\Cref{cor:EF:distToClique:agents:FPT}).

The third dimension restricts the {\em preferences} of the agents. We view the preferences as a directed graph~$\Gpref$, where an arc from~$i$ to~$j$ denotes that agent~$i$ is interested in their distance from agent~$j$.
We prove that the problem is polynomial-time solvable if~$\Gpref$ is either an in-star (\Cref{thm:EF:Gpref:instar:poly}), or an out-star (\Cref{thm:EF:Gpref:outStar:poly}). 
Interestingly, the in-star case -- which at first glance seems like an ``easy'' case -- requires a technical, tedious algorithm that is based on dynamic programming.
Our last result for this dimension considers the case where only few agents (out of the many) have preferences, i.e., there are only a few vertices in~$\Gpref$ with non-zero degree.
Then, we prove that the problem is \Wh parameterized by non-zero degree vertices and we complement this by an \XP algorithm (\Cref{thm:EF:Gpref:nonZeroDegree:poly}).

Our last set of results considers jump and swap stability. 
First, we prove that stable allocations fail to exist even on paths with two agents for jump stability (\Cref{prop:jumpSwap:twoAgents:noExist}), and three agents for swap stability (\Cref{prop:swap:notExist}). 
We complement this by showing two positive results.
Under symmetric preferences, both jump and swap stable allocations do always exist and, in addition, they can be computed in polynomial time (\Cref{thm:swapJump:symmetric:polytime}). 
We note here that~\citet{AzizCLNW2025}, using similar arguments, proved in their setting that although jump stable allocations exist for the unit interval, it is \PLS-complete to find one; it is the graph structure that allows for tractability.
Finally, if~$\Gpref$ is acyclic, then again both jump and swap stable allocations always exist and they can be computed in polynomial time (\Cref{thm:swapJump:acyclic:polytime}). Again, this result follows the same lines as~\citet{AzizCLNW2025}, who proved a similar result for jump stable allocations in their setting.

\subsection{Related Work}%

In recent years, finding stable and fair allocations of agents to a topology has emerged as an important research direction in algorithmic game theory and computational social choice literature, capturing various real-life problems.

One of the best-known models in this line of research is the so-called \emph{Schelling games on graphs}~\citep{ChauhanLM2018,AgarwalEGISV2021,BullingerSV2021,BiloBLM2022,KanellopoulosKV2022,BiloBDLMS2023,FriedrichLMS2023,DeligkasEG2024,KreiselBFN2024}, which turned out to be an important tool in the study of segregation. Here, we are given a topology and a set of agents partitioned into~$k$ colors, and the goal is to allocate these agents to the topology in a stable way. Similarly to our work, the notions of stability studied are swap and jump stability. The crucial difference from our model is that the preferences are inherent and the agents are interested only in the number of agents of their own color in their direct neighborhood (resp. in the diversity~\citep{NarayananSV2025} or variety~\citep{NarayananOTV2025} of their neighborhood).

In the \emph{seat arrangement} problems~\citep{BodlaenderHJOOZ2020,Wilczynski2023,CeylanCR2023,BerriaudCW2023,AzizLSV2025}, we are given a social network over the set of agent and the so-called seating graph, and the goal is to allocate the agents to vertices of the seating graph so that the allocation satisfies some desirable stability notion; importantly, envy-freeness, swap stability, and jump stability. However, similarly to Schelling games, agents are only interested in their friends (according to the social network) allocated to their direct neighborhood in the seating graph.

In the \emph{refugee housing} problem~\citep{KnopS2023,Schierreich2023,Schierreich2024,LisowskiS2025}, we have a topology, a set of inhabitants, and a set of refugees. The inhabitants are already occupying some vertices (or houses) of the topology, and the goal is to allocate the refugees to the empty houses so that the preferences of both the inhabitants and the refugees are respected. Importantly, the preferences in refugee housing are not based on the distances; the agents are interested only in the agents (of the opposite type) allocated to their direct neighborhood, and positions of some agents (the inhabitants) are already fixed.%

In all three previously discussed models, agents were only interested in agents allocated to their direct neighborhoods. This is not the case in the \emph{topological distance games (TDGs)}~\citep{BullingerS2023,DeligkasEKS2024b}, where we are also given a topology and a set of agents that must be allocated to the topology. In the allocation, the utility of an agent depends on both the agent's inherent utilities for other agents and its distance from them on the topology. That is, in TDGs, agents are interested in the whole allocation, not only their neighbors. However, they cannot express ideal distances to other agents as in our model.

Finally, preferences based on distances between agents also appear in other contexts in collective decision making. One such example is \emph{social distance games}~\citep{BranzeiL2011,KaklamanisKP2018,BalliuFMO2019,BalliuFMO2022,GanianHKRSS2023}, a class of \emph{hedonic games}, where we are given a social network over a set of agents, and the goal is to partition these agents into groups called coalitions. The utility of agents is computed as a (weighted) distance to all other agents in the same coalition, and the stability requirements usually arise from coalition formation literature and are very different from ours.

\section{The Model}\label{sec:model}

We use~$\N$ to denote the set of all positive integers. For any ~$i\in\N$, we denote~$[i] := \{1,\ldots,i\}$ and set~$[i]_0 = [i] \cup \{0\}$.

A \emph{distance preservation game}\footnote{From now on, we drop `graphical' from the name of our model, as it is shorter and there is no possibility of confusion.} (DPG) is a quadruplet~$\Gamma=(\agents,(\relSet_a)_{a\in\agents},(\dFn_a)_{a\in\agents},G)$, where~$\agents$ is a set of \emph{agents},~$\relSet_a \subseteq \agents\setminus\{a\}$ is a (possibly empty) \emph{relationship set} containing all ``interesting agents'' from the perspective of agent~$a\in\agents$,~$\dFn_a\colon \relSet_a\to\N$ is a \emph{distance function} assigning to each~$b\in \relSet_a$ the ideal distance agent~$a$ wants to be from~$b$, and~$G$ is a simple, undirected, and connected graph called \emph{topology} with~$|V| \geq \numAgents$. Based on agents' relationship sets, we define a \emph{preference graph}~$\Gpref$, which is a directed graph over~$\agents$ with an arc~$(a,b)$,~$a,b\in\agents$, if and only if~$b\in \relSet_a$. If for every pair of agents~$a,b\in\agents$ such that~$b \in \relSet_a$ it holds that also~$a \in \relSet_b$ and~$\dFn_a(b) = \dFn_b(a)$, we say that the preferences are \emph{symmetric}. Moreover, if~$\relSet_a = \emptyset$ for some agent~$a\in\agents$, we say that~$a$ is \emph{indifferent}.

The ultimate goal is to find an \emph{allocation}~$\alloc\colon\agents\to V$, which is an injective mapping between the set of agents and the set of vertices of the topology. We use~$V^\alloc$ to denote the set of all vertices used by the allocation~$\alloc$, and we say that a vertex~$v \in V$ is \emph{empty} if~$v\not\in V^\alloc$ and \emph{occupied} otherwise. Let~$\alloc$ be an allocation and~$a,b\in\agents$ be a pair of distinct agents. A vertex~$v\in V^\alloc$ such that~$\alloc(a) = v$ is called a \emph{location} of agent~$a$. We use~$\alloc^{a\leftrightarrow b}$ to denote the allocation in which agents~$a$ and~$b$ swapped their positions, and all other agents remain allocated to the same vertices. Similarly, for an agent~$a\in\agents$ and an empty vertex~$v\in V\setminus V^\alloc$, we use~$\alloc^{a\mapsto v}$ to denote the allocation where the agent~$a$ is allocated to the vertex~$v$ and the positions of all other agents remain the same.

Naturally, not all allocations are equally good with respect to agents' relationship sets and distance functions. Therefore, we define and study several axioms capturing conditions under which, if met, the agents consider an allocation acceptable. Let~$a\in\agents$ be an agent and~$\alloc$ be an allocation. Then, the cost for agent~$a$ in the allocation~$\alloc$, denoted~$\cost(a,\alloc)$, is defined as
\[
    \cost(a,\alloc) = \sum_{b\in \relSet_a} \Big|\dFn_a(b) - \dist_G(\alloc(a),\alloc(b))\Big|\,,
\]
where~$\dist_G(u,v)$ is the length of a shortest path between the vertices~$u$ and~$v$ in the topology~$G$. Intuitively, the cost agent~$a$ suffers from agent~$b$ is equal to the difference between the distance prescribed by the distance function and the actual distance of these two agents with respect to the allocation~$\alloc$.

First, we are interested in \emph{fair} allocations, where the agents are not envious of the positions of other agents. Specifically, we adapt the well-known notion of envy-freeness~\citep{Foley1967} from the fair division literature to our setting. Formally, envy-freeness in our context is defined as follows.

\begin{definition}%
    An allocation~$\alloc$ is called \emph{envy-free} (EF) if there is no pair of agents~$a,b\in\agents$ such that \[\cost(a,\alloc) > \cost(a,\alloc^{a\leftrightarrow b})\,.\] If~$\cost(a,\alloc) > \cost(a,\alloc^{a\leftrightarrow b})$, then we say that~$a$ \emph{envies}~$b$.
\end{definition}

In other words, envy-freeness requires that no agent can benefit from swapping positions with another agent. Note that in the definition of envy-freeness, we could also compare~$\cost(a,\alloc)$ with the cost of a in~$\alloc^{a\mapsto \alloc(b)}$, assuming that we removed~$b$ from the instance. However, since the position of~$b$ affects the cost of~$a$, and therefore~$b$ can be seen as \emph{externality} for~$a$, we define envy-freeness with respect to swaps to be in line with similar models~\citep{Velez2016,AzizSSW2023,DeligkasEKS2024}.

If an agent~$a$ envies an agent~$b$, then the corresponding allocation is not considered acceptable regardless of the impact of the swap on the cost of agent~$b$. In certain scenarios, agents are \emph{altruistic} and care about the costs for other agents. Therefore, we define a weaker notion of stability based on the swapping of two agents.

\begin{definition}%
    An allocation~$\alloc$ is called \emph{swap stable} if there is no pair of distinct agents~$a,b\in\agents$ such that
    \begin{equation*}
        \cost(a,\alloc) > \cost(a,\alloc^{a\leftrightarrow b})\quad\text{and}\quad\cost(b,\alloc) > \cost(b,\alloc^{b\leftrightarrow a})\,.
    \end{equation*}
    If a pair of agents~$a,b$ can improve their cost by swapping positions, we say that~$a$ and~$b$ admit \emph{swap deviation}.
\end{definition}

Not surprisingly, there is a certain relation between envy-freeness and swap stability.

\begin{observation}
    If an allocation~$\alloc$ is envy-free, then it is also swap stable.
\end{observation}
\begin{proof}
    Let~$\alloc$ be an envy-free allocation. For the sake of contradiction, let~$a,b\in\agents$ be a pair of distinct agents that admit a swap deviation, that is, we have~$\cost(a,\alloc) > \cost(a,\alloc^{a\leftrightarrow b})$ and~$\cost(b,\alloc) > \cost(b,\alloc^{b\leftrightarrow a})$. Clearly, agent~$a$ envies~$b$ and vice versa, which contradicts that~$\alloc$ is envy-free.
\end{proof}

In the opposite direction, not all swap stable allocations are envy-free. We give a formal argument for this claim later in \Cref{thm:swapAndNotEF}.

Observe that in  envy-free (or swap stable) allocations, agents are envious of the positions \emph{of other agents}. In particular, an allocation is considered fair even though there is an agent~$a\in\agents$ and an \emph{empty} vertex~$v\in V$ such that~$\cost(a,\alloc) > \cost(a,\alloc^{a \mapsto v})$. One can argue that even though such an allocation is formally fair, selfish agents would still have an incentive to move to a better location. Motivated by this, we define \emph{jump stability}.

\begin{definition}%
    An allocation~$\alloc$ is called \emph{jump stable} if there is no agent~$a\in\agents$ and an empty vertex~$v\in V$ such that \[
        \cost(a,\alloc) > \cost(a,\alloc^{a\mapsto v})\,.
    \]
    If an agent~$a$ can improve its cost by jumping to an empty vertex, then we say that~$a$ admits a \emph{jump deviation}.
\end{definition}

Naturally, we can combine envy-freeness with jump stability. However, as we show in the following proposition, then the problem reduces to a simple decision of the existence of envy-free allocations for a carefully constructed instance with no empty vertex.

\begin{proposition}
    A distance preservation game~$\Gamma = (\agents,(\relSet_a)_{a\in\agents},(\dFn_a)_{a\in\agents},G)$ admits an envy-free and jump stable allocation~$\alloc$ if and only if the distance preservation game~$\Gamma' = (A',(\relSet'_a)_{a\in\agents'},(\dFn'_a)_{a\in\agents'},G)$, where~$\agents' = \agents \cup \{ x_i \mid i \in [|V|-|A|] \}$ and~$\relSet'_a = \relSet_a$ and~$\dFn'_a = \dFn_a$ whenever~$a\in\agents$ and~$\relSet'_a = \dFn'_a = \emptyset$, otherwise, admits an envy-free allocation~$\alloc'$.
\end{proposition}
\begin{proof}
    Let~$\alloc$ be an envy-free and jump stable allocation for~$\Gamma$ and, for the sake of contradiction, assume that~$\alloc'$, which is the same as~$\alloc$ for agents of~$\agents$ and allocates the added agents~$x_i$ arbitrarily to empty vertices of~$G$, is not an envy-free allocation for~$\Gamma'$. Clearly, no added agent~$x_i$ is envious, as his relationship set is empty, and therefore his cost is always zero. Hence, the envious agent in~$\Gamma$ is some~$a\in\agents$. However, since~$\alloc$ is  envy-free for~$\Gamma$ and the topology is the same in both~$\Gamma$ and~$\Gamma'$,~$a$ cannot envy some agent~$b\in\agents\setminus\{b\}$. Therefore,~$a$ must be envious of some added agent~$x_i$,~$i\in[|V|-\numAgents]$; that is, it holds that~$\cost(a,\alloc') > \cost(a,\alloc'^{a\leftrightarrow x_i})$. However, since~$\alloc'(x_i)$ is an empty vertex according to~$\alloc$, it also holds that~$\cost(a,\alloc) > \cost(a,\alloc^{a\mapsto \alloc'(x_i)})$, which contradicts that~$\alloc$ is jump stable.

    In the opposite direction, let~$\alloc'$ be an envy-free allocation for~$\Gamma$ and let~$\alloc$ be an allocation~$\alloc'$ with the domain narrowed to only~$\agents$. We claim that~$\alloc$ is an envy-free and jump stable allocation for~$\Gamma$. Envy-freeness directly follows from the fact that if we remove an agent from an instance, we cannot introduce new envy. It remains to show that~$\alloc$ is jump stable. For the sake of contradiction, let there be~$a\in\agents$ and an empty vertex~$v\in V\setminus V^\alloc$ so that~$\cost(a,\alloc) > \cost(a,\alloc^{a\mapsto v})$. Then, however, in~$\Gamma'$, vertex~$v$ is occupied by some added agent~$x_i$, meaning that~$\cost(a,\alloc') > \cost(a,\alloc'^{a \leftrightarrow x_i})$, which contradicts that~$\alloc'$ is  envy-free for~$\Gamma'$.
\end{proof}

To illustrate our model and the studied stability notions, we conclude with a running example.

\begin{example}
    Let us have an instance with three agents~$a$,~$b$, and~$c$, and let the topology~$G$, the preference graph~$\Gpref$, and the allocation~$\alloc$ be as depicted on \Cref{fig:example}.
    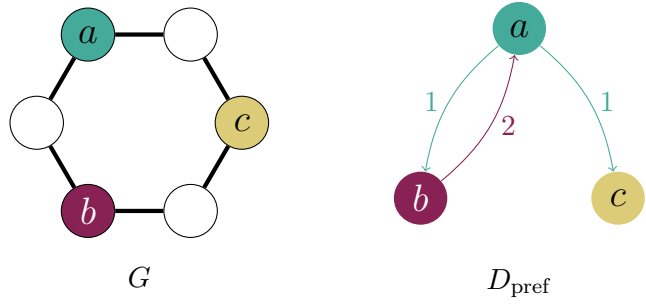
\begin{figure}[bt!]
        \centering
        \begin{tikzpicture}[scale=1.5]
            \def\radius{1cm}
            \coordinate (A) at (90:\radius);
            \coordinate (B) at (210:\radius);
            \coordinate (C) at (330:\radius);

            \node[circle,fill=cbGreen,inner sep=2pt,minimum width=20pt] (a) at (A) {\Large$a$};
            \node[circle,fill=cbRed,text=white,inner sep=2pt,minimum width=20pt] (b) at (B) {\Large$b$};
            \node[circle,fill=cbOrange,inner sep=2pt,minimum width=20pt] (c) at (C) {\Large$c$};
            
            \draw[->,cbGreen,bend right=20] (a) to node [midway,left] {$1$} (b);
            \draw[->,cbGreen,bend left=20] (a) to node [midway,right] {$1$} (c);
            \draw[->,cbRed,bend right=20] (b) to node [midway,right] {$2$} (a);

            \node at (0,-1.25) {$\Gpref$};
            
            \def\radius{0.9cm}

            \node[draw,circle,inner sep=2.5pt,xshift=-5cm,yshift=0.25cm,fill=cbOrange,minimum width=20pt] at ({1*60 - 60}:\radius) (N1) {\Large$c$};
            \node[draw,circle,inner sep=2.5pt,xshift=-5cm,yshift=0.25cm,minimum width=20pt] at ({2*60 - 60}:\radius) (N2) {};
            \node[draw,circle,inner sep=2.5pt,xshift=-5cm,yshift=0.25cm,fill=cbGreen,minimum width=20pt] at ({3*60 - 60}:\radius) (N3) {\Large$a$};
            \node[draw,circle,inner sep=2.5pt,xshift=-5cm,yshift=0.25cm,minimum width=20pt] at ({4*60 - 60}:\radius) (N4) {};
            \node[draw,circle,inner sep=2.5pt,xshift=-5cm,yshift=0.25cm,fill=cbRed,text=white,minimum width=20pt] at ({5*60 - 60}:\radius) (N5) {\Large$b$};
            \node[draw,circle,inner sep=2.5pt,xshift=-5cm,yshift=0.25cm,minimum width=20pt] at ({6*60 - 60}:\radius) (N6) {};
            
            \draw[ultra thick] (N1) -- (N2) -- (N3) -- (N4) -- (N5) -- (N6) -- (N1);
            \node[xshift=-3.5cm] at (-1,-1.2) {$G$};
        \end{tikzpicture}
        \caption{An example of a topology with allocated agents (on the left) and a preference graph (on the right). Observe that~$\relSet_a = \{b,c\}$,~$\relSet_b = \{a\}$, and~$\relSet_c = \emptyset$, that is, agent~$c$ is indifferent. The preferences are clearly not symmetric.}
        \label{fig:example}
    \end{figure}
    The cost of agent~$b$ is zero, as~$b$ requires to be at distance two from~$a$ and this is satisfied. The cost for agent~$c$ is always zero, since he is indifferent. The cost of agent~$a$, on the other hand, is two, as she requires to be at distance one from both remaining agents, but her distance from both of them is two.

    Regarding stability of~$\alloc$, clearly no agent~$b$ or~$c$ can participate in a swap, as they are already on their best position. Hence,~$\alloc$ is swap stable. Similarly, agent~$a$ is not envious of the position of any of the other two agents, as if she swaps with one of them, the distances are preserved. Finally, allocation~$\alloc$ is not jump stable, as agent~$a$ can jump to the common neighbor~$v$ of~$b$ and~$c$ and decrease her cost to zero. In response, agent~$b$ will jump so she is distance~$2$ from agent~$a$. It is easy to see that~$\alloc'$ is jump stable.
\end{example}

\section{Envy-Freeness}

As we showed in the previous section, envy-freeness is closely related to all other stability notions assumed in this work. Thus, we begin our investigation with this notation.

First, we study the conditions under which envy-free allocations are guaranteed to exist. It is easy to see that if there are two agents, any allocation is envy-free, as if the agents swap their positions, their distance remains the same.

\begin{proposition}
\label{prop:EF:two_agents}
    If~$\numAgents = 2$, an envy-free allocation is guaranteed to exist and can be found in constant time.
\end{proposition}
\begin{proof}
    Let~$G$ be a topology and~$\alloc$ be an arbitrary allocation of our agents. Since the graph is undirected, we have that~$\dist_G(\alloc(a_1),\alloc(a_2))$ is the same as~$\dist_G(\alloc(a_2),\alloc(a_1))$. Therefore, if the agents swap their positions, their utility remains the same. Consequently,~$\alloc$ is envy-free.
\end{proof}

However, as we show in our next result, already with three agents, envy-free allocation may not exist.

\begin{proposition}\label{thm:EF:notGuaranteed}
    For every~$\numAgents \geq 3$, an envy-free allocation is not guaranteed to exist.
\end{proposition}
\begin{proof}
    Let~$\numAgents = 3$ and let the topology~$G$ be a cycle with four vertices. Each agent~$a_i$,~$i\in[\numAgents]$, has~$M_i = \{a_{(i\bmod \numAgents)+1}\}$ and requires this agent at distance two. For the sake of contradiction, assume that there is an envy-free allocation~$\alloc$. First, let~$\dist_G(\alloc(a_1),\alloc(a_2)) = 1$. If~$a_3$ is at distance two from~$a_2$, then~$a_1$ is envious of~$a_3$'s position. Hence,~$a_3$ is at a distance exactly one from~$a_2$. However, in this case,~$a_3$ is simultaneously at distance two from~$a_1$, which means that~$a_2$ is envious of~$a_1$'s position. Thus, it must be the case that~$\dist(\alloc(a_1),\alloc(a_2)) = 2$. Regardless of the position of~$a_3$, she is always at distance one from~$a_1$ and envies agent~$a_2$. That is, no envy-free allocation exists.

    For~$\numAgents > 3$, the topology is just a cycle with~$\numAgents$ or~$\numAgents+1$ vertices, depending on whether~$\numAgents$ is even or odd, and each agent~$a_i$,~$i\in[\numAgents]$, requires the agent~$a_{(i\bmod \numAgents)+1}$ to be at distance two. Clearly, this is not satisfiable for all agents, and as almost all vertices are occupied, there is always at least one envious agent.
\end{proof}

With the previous no-instance in hand, we finally give a formal proof that there are games admitting swap stable but not envy-free allocations, as promised in \Cref{sec:model}.

\begin{proposition}\label{thm:swapAndNotEF}
    There is a distance preservation game~$\Gamma$ that admits a swap stable allocation but no envy-free allocation.
\end{proposition}
\begin{proof}
    Recall the instance from \Cref{thm:EF:notGuaranteed} with~$\numAgents = 3$. It was shown that such a game does not admit envy-free allocation. However, assume an allocation with~$\alloc(a_1) = v_1$,~$\alloc(a_3) = v_2$, and~$\alloc(a_2) = v_3$ (we suppose that~$v_1,\ldots,v_5$ is a DFS order if we start with the vertex~$v_1$). Then, the cost for the agent~$a_1$ is zero, and the costs for both~$a_2$ and~$a_3$ are one. Therefore,~$a_1$ participates in no swap deviation. Although the cost for~$a_3$ in~$\alloc^{a_3\leftrightarrow a_2}$ is zero, that is, decreased, the agent~$a_2$ is in~$\alloc^{a_3\leftrightarrow a_2}$ still at distance one from~$a_3$, meaning that the swap is not beneficial to him. That is, the allocation~$\alloc$ is swap stable.
\end{proof}

Since there are instances with no envy-free allocation, it is natural to ask, given a distance preservation game~$\Gamma$, how hard it is to decide whether~$\Gamma$ admits an  envy-free allocation. In our first hardness result, we show that this problem is \NPc, even if the preferences of our agents are symmetric, i.e., if an agent~$a$ requires agent~$b$ to be at distance~$\delta$, then also agent~$b$ requires agent~$a$ to be at the same distance.

\begin{theorem}\label{thm:EF:NPh}
    It is \NPc to decide whether a distance preservation game~$\Gamma$ admits an envy-free allocation, even if the preferences are symmetric and the domain of each~$\dFn_a$,~$a\in\agents$, is~$\{1\}$.
\end{theorem}
\begin{proof}
    We reduce from the \probName{$3$-Partition} problem. In this problem, we are given a multi-set~$S=\{s_1,\ldots,s_{3N}\}$ of integers such that~$\sum_{i\in[3N]} s_i = N\cdot B$, and the goal is to decide whether a set~$\mathcal{X}$ of~$N$ pairwise disjoint~$3$-sized subsets of~$S$ exists so that the elements of every~$X\in\mathcal{X}$ sum up to exactly~$B$. The problem is known to be \NPc even if~$B/4 < s < B/2$ for every~$s\in S$~\citep{GareyJ1975}.

    Given an instance~$S$ of \probName{$3$-Partition}, we construct an equivalent distance preservation game~$\Gamma$ as follows. The topology~$G$ consists of a disjoint union of~$N$ cliques~$C_1,\ldots,C_N$, each of size exactly~$B$, and one apex vertex~$v_g$ connected to all other vertices. The set of agents contains exactly~$s_i$ \emph{item agents}~$a_i^1,\ldots,a_i^{s_i}$ for every~$s_i\in S$, and one \emph{guard agent}~$g$. For every item agent~$a_i^j$, we have~$M_{a_i^j} = \{ a_i^1,\ldots,a_i^{s_i} \}\setminus \{a_i^j\}$. For the guard agent, the set~$M_{g}$ contains all the other agents. Each agent requires to be at distance one from all other agents he cares about.

    For correctness, assume first that~$S$ is a \Yes-instance and~$\mathcal{X} = (X_1,\ldots,X_N)$ is a solution partition. We construct the following allocation~$\alloc$ and claim that it is envy-free. First, we set~$\alloc(g) = v_g$. Then, let~$X_j = \{s_{i_1},s_{i_2},s_{i_3}\}$ be a set of~$\mathcal{X}$. For every item agent~$a_i^j$, where~$i\in{i_1,i_2,i_3}$ and~$j\in[s_i]$, we arbitrarily allocate this agent to an empty vertex of the clique~$C_j$. Note that since~$\mathcal{X}$ is a solution, it must hold that~$s_{i_1}+s_{i_2} + s_{i_3} = B$. Since~$V(C_j) = B$, such an empty vertex always exists. To show that~$\alloc$ is envy-free, we just observe that the cost of all agents is zero: The guard is at distance one from all other agents and the item agents are in the same clique as all other agents corresponding to the same item.

    In the opposite direction, let~$\alloc$ be an envy-free allocation for~$\Gamma$. First, observe that if the guard agent is not allocated to~$v_g$, then~$g$ is envious towards~$\alloc^{-1}(v_g)$, and at least one agent is allocated to~$v_g$ as the number of agents is equal to the number of vertices of~$G$. Next, we show that all item agents corresponding to the same~$s_i\in S$ are allocated to the same clique~$C_j$.

    \begin{claim}
        For every~$i\in[3N]$ there exists one~$\ell\in[N]$ such that~$\alloc(a_i^j) \in V(C_\ell)$ for every~$j\in[s_i]$.
    \end{claim}
    \begin{claimproof}
        For the sake of contradiction, assume that~$\alloc$ is envy-free and the statement is not true. Without loss of generality, let~$C_1$ be a clique with the largest number of agents of~$\{a_i^1,\ldots,a_i^{s_i}\}$ allocated to it. Let~$a_i^j$ be an agent so that~$\alloc(a_i^j)\not\in V(C_1)$. Since~$\alloc(a_i^j)\not= v_g$ by the previous arguments, it must be the case that~$\alloc(a_i^j)\in V(C_\ell)$ for some~$\ell\in[2,N]$. The cost of agent~$a_i^j$ in~$\alloc$ is~$s_i - |\{ a_i^{j'} \mid j'\in[s_i]\setminus\{j\} \land \alloc(a_i^{j'}) \in V(C_\ell) \}|$. Let~$b$ be an agent with~$\alloc(b)\in V(C_1)$ and~$b\not\in M_{a_i^j}$ (such an agent always exists since each~$s_i$ is smaller than~$B/2$). In an allocation~$\alloc^{a_i^j\leftrightarrow b}$, the cost for agent~$a_i^j$ is~$s_i - |\{a_i^{j'} \mid j'\in[s_i]\setminus\{j\} \land \alloc(a_i^{j'}) \in V(C_1) \}|$, which is by at least one smaller than~$s_i - |\{ a_i^{j'} \mid j'\in[s_i]\setminus\{j\} \land \alloc(a_i^{j'}) \in V(C_\ell) \}|$ since~$C_1$ contains the largest number of agents corresponding to~$s_i$. That is, agent~$a_i^j$ has envy towards agent~$b$, which contradicts that~$\alloc$ is envy-free. Hence, all item agents corresponding to the same item~$s_i$ are allocated to the same clique.
    \end{claimproof}

    Now, from the fact that~$g$ is allocated to~$v_g$, that all item agents corresponding to the same~$s_i$ are part of the same clique, that each clique is of size exactly~$B$, and that for every~$s_i$ we have more than~$B/4$ item agents, we directly obtain that if we take~$\mathcal{X} = (X_1,\ldots,X_N)$ so that~$X_j = \{ s_i \mid \alloc(a_i^1) \in V(C_j) \}$, we obtain a solution for~$S$. This finishes the proof.
\end{proof}

\subsection{Restricted Topology}\label{sec:EF:topology}

The first restriction we study is the restriction of the underlying topology to which we allocate. It is reasonable to assume that if the topology is well-structured, then the problem should become easier. And indeed, in our first result, we support this intuition with two polynomial-time solvable cases. Specifically, if the topology is a star graph or a clique, efficient algorithms are possible.

\begin{theorem}
\label{thm:EF:star_clique}
    If the topology~$G$ is a star graph or a clique, the existence of an envy-free allocation can be decided in polynomial time.
\end{theorem}
\begin{proof}
    If the topology is a clique, then all distances are the same. Hence, an arbitrary allocation is envy-free, as agents cannot change their distance to other agents by swapping positions.

    When the topology is a star, we have~$\numAgents + 1$ possibilities for~$\alloc^{-1}(c)$, where~$c$ is the center of~$G$: either some agent~$a\in\agents$ is allocated to~$c$ or it is empty. As the leaves are symmetric with respect to the distances to other vertices, we can arbitrarily allocate the remaining agents to the leaves for every choice of~$\alloc^{-1}(c)$ and check in polynomial time whether at least one of these allocations is envy-free.
\end{proof}

However, as we demonstrate in the remainder of this subsection, the tractability boundary cannot be pushed much further. Observe that stars are the only family of connected graphs with the vertex cover number one. In the following, we show that on topologies with the vertex cover number two, deciding the existence of an envy-free allocation becomes computationally intractable.

\begin{theorem}\label{thm:EF:pNPh:vc}
    It is \NPc to decide whether a distance preservation game~$\Gamma$ admits an envy-free allocation, even if the topology~$G$ is of the vertex cover number~$2$.
\end{theorem}
\begin{proof}
    We give a polynomial reduction from the \probName{Cubic Bisection} problem. Here, we are given a cubic graph~$H$ with~$2N$ vertices and an integer~$k\in\N$, and the goal is to decide whether a partition~$(X,Y)$ of~$V(H)$ exists such that~$|X| = |Y| = N$ and there are at most~$k$ edges with one endpoint in~$X$ and one endpoint in~$Y$. The problem is \NPc even if there is no partition of~$V(H)$ to two equal parts with less than~$k$ edges between~$X$ and~$Y$, and~$N > k > 4$~\cite[Theorem 3]{DeligkasEKS2024}. Also note that since~$H$ is~$3$-regular, then in every partition~$(X,Y)$ with exactly~$k$ edges between~$X$ and~$Y$ we have exactly~$B = \frac{3N-k}{2}$ edges with both endpoints in the same part. Consequently, we can assume that~$\frac{3N-k}{2}$ is an integer, as otherwise the \probName{Cubic Bisection} instance is trivially a \No-instance.

    \begin{figure}[bt!]
        \centering
        \begin{tikzpicture}
            \node[draw,circle,inner sep=2pt,very thick] (S1) at (1.5,0) {$c_1$};
            \node[draw,circle,inner sep=2pt,very thick] (S2) at (4.5,0) {$c_2$};
            
            \foreach[count=\i] \y in {1.5,1,0.5,-0.5,-1,-1.5} {
                \node[draw,circle,inner sep=1.5pt] (A\i) at (0,\y) {};
                \draw (S1) -- (A\i);
            }
            \node (Adummy) at (0,0.1) {$\vdots$};
    
            \foreach[count=\i] \y in {1,0.5,-0.5,-1} {
                \node[draw,circle,inner sep=1.5pt] (k\i) at (3,\y) {};
                \draw (k\i) edge (S1) edge (S2);
            }
            \node (kdummy) at (3,0.1) {$\vdots$};
            \node (kdummy) at (3,-1.75) {$M$};
            
            \foreach[count=\i] \y in {1.5,1,0.5,-0.5,-1,-1.5} {
                \node[draw,circle,inner sep=1.5pt] (B\i) at (6,\y) {};
                \draw (S2) -- (B\i);
            }
            \node (Bdummy) at (6,0.1) {$\vdots$};

            \begin{pgfonlayer}{background}
                \draw[fill,cbOrange,rounded corners] (-0.5,2) rectangle (0.5,-2); 
                \draw[fill,cbGreen,rounded corners] (2.5,1.5) rectangle (3.5,-1.5); 
                \draw[fill,,cbOrange,rounded corners] ( 5.5,2) rectangle (6.5,-2); 
            \end{pgfonlayer}
        \end{tikzpicture}
        \caption{An illustration of the topology used in the proof of \Cref{thm:EF:pNPh:vc}. Each set of vertices with yellow (light gray) background is of size~$B=\frac{3N-k}{2}$ and the set~$M$ with green (dark gray) background is of size~$k$. Clearly, the vertices~$c_1$ and~$c_2$ form a vertex cover of this graph of size two.}
        \label{fig:EF:pNPh:vc}
    \end{figure}
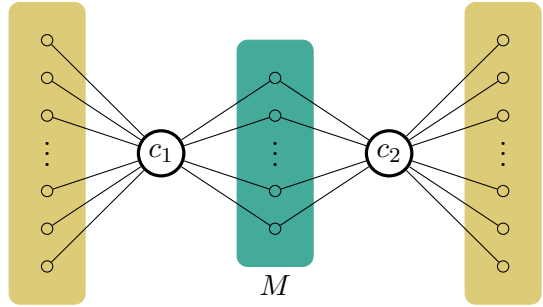

    Let~$\mathcal{I} = (H,k)$ be an instance of the \probName{Cubic Bisection} problem. We construct an equivalent distance preservation game~$\Gamma$ as follows. The topology~$G$ consists of two stars~$S_1$ and~$S_2$ with exactly~$B$ leaves each. Let~$c_1$ and~$c_2$ be centers of~$S_1$ and~$S_2$, respectively. We additionally add~$k$ parallel edges connecting~$c_1$ and~$c_2$ and subdivide each such edge once. We set~$M$ to be the set of vertices resulting from the subdivision. See \Cref{fig:EF:pNPh:vc} for an illustration of the topology. The set of agents~$\agents$ contains one \emph{edge agent}~$a_{u,v}$ for every pair of~$u,v\in V(H)$ such that~$\{u,v\}\in E(H)$; i.e., there is one agent for each edge of~$H$. Additionally, we have two \emph{guard agents}~$g_1$ and~$g_2$. Observe that the number of agents is equal to the number of vertices of~$G$. The preferences of our agents are as follows. For each agent~$a_{u,v}$, we have~$M_{a_{u,v}} = \{ a_{w,w'} \mid |\{u,v\} \cap \{w,w'\}| > 0 \}$; that is, each agent cares about agents representing adjacent edges. For every guard agent~$g_i$ we set~$M_{g_i} = \agents\setminus\{g_1,g_2\}$. The distance edge agents require is always two, while the distance required by the guard agents is always one.
    
    For correctness, let~$\mathcal{I}$ be a \Yes-instance and~$(X,Y)$ be a solution partition. We construct an envy-free allocation for~$\Gamma$ as follows. For every~$\{u,v\} \in E(H)$ we allocate~$a_{u,v}$ to
    \begin{itemize}
        \item an arbitrary empty leaf of~$S_1$ if~$u,v\in X$,
        \item an arbitrary empty leaf of~$S_2$ if~$u,v\in Y$, or
        \item an arbitrary empty vertex of~$M$ otherwise.
    \end{itemize}
    Observe that since~$(X,Y)$ is a bipartition, exactly~$|M| = k$ edge agents are allocated to the vertices of~$M$ and exactly~$B$ agents are allocated to leaves of~$S_1$ and~$S_2$, respectively. Finally, we set~$\alloc(g_1) = c_1$ and~$\alloc(g_2) = c_2$.
    
    Now, we claim that~$\alloc$ is an envy-free allocation. First, every guard agent~$g_i$,~$i\in[2]$, has exactly~$B+k$ edge agents at distance one, and~$B$ edge agents at distance three---the cost for~$g_i$ in~$\alloc$ is~$2B$. Hence, they are not envious of each other. Assume that~$g_i$ envies some element agent~$a$ allocated to a leaf of~$G$. Then, after the swap, at most one edge agent is at distance one, and all other edge agents are at distance at least two. Therefore, the cost for~$g_i$ after swapping is at least~$2B+k-1$, which is larger than~$2B$. The same holds if~$g_i$ were to swap with some agent allocated to~$M$, so the allocation is envy-free for the guard agents. Let~$a_{u,v}$ be an edge agent. We claim that all agents of~$M_{a_{u,v}}$ are at distance two, and therefore, the cost of~$a_{u,v}$ in~$\alloc$ is zero. If~$\alloc(a_{u,v}) \in M$, then it is clearly the case as all the vertices except~$c_1$ and~$c_2$ are at distance two from~$\alloc(a_{u,v})$ and, moreover,~$c_1$ and~$c_2$ are occupied by~$g_1$ and~$g_2$, respectively. Hence,~$a_{u,v}\not\in M$. Without loss of generality, let~$a_{u,v}$ be allocated to a leaf of~$S_1$ and, for the sake of contradiction, let there be an agent~$a_{u,w}$ allocated to a leaf of~$S_2$. Since~$a_{u,v}\in V(S_1)$, then by the definition of the allocation, it must be the case that~$u,v\in X$. By the same argument, since~$a_{u,w}\in V(S_2)$, it must be the case that~$u,w\in Y$. This contradicts that~$(X,Y)$ is a partition of~$V(H)$. Hence, such a situation never occurs, and indeed, the cost of every edge agent is zero, so they are not envious.%

    In the opposite direction, let~$\alloc$ be an envy-free allocation for~$\Gamma$. First, we show that the guard agents are allocated to the star centers. 

    \begin{claim}\label{thm:EF:pNPh:vc:guardAgents}
        We have~$\{ \alloc(g_1), \alloc(g_2) \} = \{ c_1, c_2 \}$.
    \end{claim}
    \begin{claimproof}
        For the sake of contradiction, assume that~$\alloc(g)\notin \{ c_1, c_2 \}$ for some~$g\in \{g_1, g_2\}$. Then, there is at most one edge agent at distance one from~$g$, and the remaining edge agents are at distance at least two, hence the cost for~$g$ is at least~$2B+k-1$. On the other hand, there are at least~$B+k-2$ edge agents at distance one from~$c_1$, at most one agent at distance two, and at most~$B$ many agents at distance three from~$c_1$. Hence, the cost for~$g$ after swapping to~$c_1$ would be at most~$2B+1 < 2B+k-1$, since~$k > 4$. It follows that if agent~$\{ \alloc(g_1), \alloc(g_2) \} \neq \{ c_1, c_2 \}$, then the agent~$g\in \{g_1, g_2\}$ such that~$g$ is not allocated to~$c_1$ nor~$c_2$ envies the agent allocated to~$c_1$. 
    \end{claimproof}

    Note that guard agents have the same preferences, and edge agents do not care about them. Therefore, if~$\alloc(g_1) = c_2$ and~$\alloc(g_2) = c_1$, then also~$\alloc^{g_1\leftrightarrow g_2}$ is an envy-free allocation. Thus, for the rest of the proof, we assume that~$\alloc(g_1) = c_1$ and~$\alloc(g_2) = c_2$.

    Next, we show that edge agents representing edges that share an endpoint cannot be allocated one to a leaf of~$S_1$ and the other to a leaf of~$S_2$.
    
    \begin{claim}\label{thm:EF:pNPh:vc:edgeAgentsInM}
        Let~$v\in V(H)$ and~$\agents_v = \{a_{uw}\mid v\in \{u,w\}\}$ be the set of agents associated with the edges incident on~$v$. Then either~$\alloc(\agents_v)\subseteq V(S_1)\cup M$ or~$\alloc(\agents_v)\subseteq V(S_2)\cup M$. 
    \end{claim}
    \begin{claimproof}
        For the sake of contradiction, assume that for some~$v\in V(H)$, there is an edge agent~$a_{u,v}$ with~$\alloc(a_{u,v})\in V(S_1)$ and an edge agent~$a_{v,w}$ with~$\alloc(a_{v,w})\in V(S_2)$. Note that~$a_{v,w}\in \relSet_{a_{u,v}}$ as the two edges share the common endpoint~$v$ (the same holds for any two agents in~$\agents_v$).  Since by \Cref{thm:EF:pNPh:vc:guardAgents},~$\{ \alloc(g_1), \alloc(g_2) \} = \{ c_1, c_2 \}$, it follows that all edge agents are allocated to~$V(G)\setminus \{c_1,c_2\}$. Therefore,~$\alloc(a_{u,v})$ is a leaf in~$S_1$ and~$\alloc(a_{v,w})$ is a leaf in~$S_2$ and their distance is~$4$. Hence~$\cost(a_{u,v},\pi) \ge 2$. On the other hand, every vertex~$m\in M$ is at distance two to any other vertex in~$V(G)\setminus \{c_1,c_2\}$. Hence, if~$b$ is an agent with~$\alloc(b)\in M$, then~$\cost(a_{u,v},\alloc^{a_{u,v}\leftrightarrow b}) = 0$ and agent~$a_{u,v}$ envies agent~$b$.
    \end{claimproof}

    Now let~$X = \{v\in V(H)\mid \alloc(\agents_v)\cap V(S_1)\neq \emptyset\}$ be the set of vertices~$v\in V(H)$ with some incident edge allocated to a leaf of~$S_1$. Similarly, let~$Y = \{v\in V(H)\mid \alloc(\agents_v)\cap V(S_2)\neq \emptyset\}$ be the set of vertices~$v\in V(H)$ with some incident edge allocated to a leaf of~$S_2$. By \Cref{thm:EF:pNPh:vc:edgeAgentsInM},~$X$ and~$Y$ are disjoint, so~$|X|+|Y|\le |V(H)| = 2N$ and at least one of the two sets has size at most~$N$. Without loss of generality, let us assume~$|X|\le N$ (the case when~$|X| > N$ and~$|Y|< N$ is analogous).
    By \Cref{thm:EF:pNPh:vc:edgeAgentsInM}, agents in~$\agents_X = \bigcup_{v\in X}\agents_v$ are all allocated on vertices in~$M\cup V(S_1)$.
    Let~$F = \{uv\in E(H)\mid \alloc(a_{u,v})\in M\text{ and }\{u,v\}\cap X\neq \emptyset\}$ be the agents in~$\agents_X$ that are allocated to vertices in~$M$. Moreover, let~$F=F_1\cup F_2$, where edges in~$F_2$ have both endpoints in~$X$ and edges in~$F_1$ have exactly one endpoint in~$X$. Since~$H$ is cubic, it follows that~$3|X| = 2B+2|F_2|+|F_1|= 3N-k+2|F_2|+|F_1|$. It follows that~$|X| = N-\frac{k-2|F_2|-|F_1|}{3}$. Let~$X'$ be an arbitrary set of vertices in~$V(H)\setminus X$ of size~$\frac{k-2|F_2|-|F_1|}{3}$ (recall that~$|X|\le N$, so~$|X'|= N-|X|$ is nonnegative). Clearly,~$(X\cup X', V(H)\setminus (X\cup X'))$ is a partition of~$V(H)$ with~$|X\cup X'| = |V(H)\setminus(X\cup X')| = N$. By \Cref{thm:EF:pNPh:vc:edgeAgentsInM} and choice of~$X$, there are at most~$|F_1|\le |F|\le |M|$ many edges with one endpoint in~$X$ and the other in~$V(H)\setminus(X\cup X')$. Finally, since~$H$ is cubic, there are at most~$3|X'|= k-2|F_2|-|F_1|$ edges incident on~$X'$. Hence, there are at most~$|F_1|+k-2|F_2|-|F_1|=k-2|F_2|$ many edges with one endpoint in~$X\cup X'$ and the other in~$V(H)\setminus(X\cup X')$ and~$(H,k)$ is a Yes-instance of  \probName{Cubic Bisection}.%
\end{proof}

We can also try to generalize stars to general trees, which are not ruled out by the previous hardness result. Nevertheless, this is not possible, as we show in the next construction. Observe that the trees used in this construction are of depth two, and, at the same time, stars (rooted in their centers) are the only family of graphs of depth one, so the result is again tight. Moreover, the preferences are again symmetric.

\begin{theorem}
\label{thm:NPh:tree}
    It is \NPc to decide whether a distance preservation game admits an  envy-free allocation, even if the preferences are symmetric and the topology is a tree of depth~2.
\end{theorem}
\begin{proof}
    We reduce from \probName{3-Partition}. This reduction follows a similar idea to that of \Cref{thm:EF:NPh}. We assume that~$N > 2$, as otherwise the instance is easy, and~$B > N$, as we can multiply each element by large enough constant to obtain an equivalent instance.

    Given an instance~$S$ of \probName{$3$-Partition}, we construct an equivalent distance preservation game~$\Gamma$ as follows. We create a topology~$G$ made up of a disjoint union of~$N$ stars~$C_1,\ldots,C_N$, each with exactly~$B$ leaves with center vertices~$u_x$ for~$x \in [1,N]$ , and one apex vertex~$v_g$ connected to all other centers of the stars; i.e.,~$G$ is a tree of depth two rooted at~$v_g$. The set of agents contains exactly~$s_i$ \emph{item agents}~$a_i^1,\ldots,a_i^{s_i}$ for every~$s_i\in S$, one \emph{guard agent}~$g$, and~$N$ \emph{special agents}~$b_1,\ldots,b_N$. For every item agent~$a_i^j$, we have~$M_{a_i^j} = \{ a_i^1,\ldots,a_i^{s_i} \}\setminus \{a_i^j\} \cup \{g\}$. For the guard agent, the set~$M_{g}$ contains all the other agents, and for each special agent~$b_j$ we have~$\relSet_{b_j} = \{g\}$. Each item agent requires to be at distance two from all other agents he cares about. Guard agent requires distance one from special agents and distance two from item agents, and special agents require to be at distance one from the guard agent. Observe that the number of agents is equal to the number of vertices of the topology.

    Assume that we are given a solution to the instance~$S$. We will construct the allocation~$\pi$ as follows. For each triplet~$X_j = \{s_{i1}, s_{i2}, s_{i3}\}$ of the partition solution, we assign all the corresponding agents to the leaves of the same star gadget~$C_j$. Since~$X_j$ is a triplet which sums up exactly to~$B$, all leaves of~$C_j$ have agents assigned to them. We allocate the guard agent~$g$ to the root~$v_g$ and we assign each~$b_j$,~$j\in[N]$, to the center of~$C_j$. Observe that all agents have cost zero in this allocation by construction. This is because all item agents are placed on the same star gadget, and hence distance two from one another. In addition, the guard agent is distance two from every item agent and in the neighborhood of each special agent. Hence~$\alloc$ is envy-free. 

    Now, assume that we have an envy-free allocation~$\alloc$ for~$\Gamma$. We show several auxiliary lemmas that help us prove that there indeed exists a solution for~$\mathcal{I}$. 
    First, we show that agent~$g$ is allocated to~$v_g$.

    \begin{claim}
        It holds that~$\alloc(g) = v_g$.
    \end{claim}
    \begin{claimproof}
        If~$\alloc(g) \not= v_g$, then~$g$ is either allocated to a leaf, or to some center~$u_j$,~$j\in[N]$.
        First, assume that~$v_g$ is allocated to a leaf; he incurs a cost of at least~$B\cdot N + N - B - 1$. However, if he swaps with the agent allocated to~$v_g$, his cost is at most~$2N$, which is strictly smaller. That is,~$g$ is envious towards~$\alloc^{-1}(v_g)$. 
        Next, let~$\alloc(g) = u_j$ for some~$j\in[N]$. Again, the cost for~$g$ is at least~$B\cdot N + N - N - N + 1$. However, by swapping with~$\alloc^{-1}(v_g)$, the cost decrease to at most~$2N$, which is again strictly smaller by our assumptions. Thus,~$g$ must be allocated to~$v_g$.
    \end{claimproof}

    Next, we show that each special agent is allocated to a center of some star.

    \begin{claim}
        For every agent~$b_i$,~$i\in[N]$, we have~$\alloc(b_i) = u_j$ for some~$j\in[N]$.
    \end{claim}
    \begin{claimproof}
        By the previous claim,~$\alloc(g) = v_g$. Hence, if a special agent~$b$ is not allocated to a center, it has to be allocated to a leaf and its cost is~$1$. Since~$\numAgents = |V(G)|$, all centers are occupied and~$b$ envies agents allocated to centers.
    \end{claimproof}

    Without loss of generality, for the rest of the proof, we can assume that each~$b_j$,~$j\in[N]$, is allocated to~$u_j$, as special agents care only about the guard agent, by the previous claim, all of them are allocated to a center, and swapping locations of two special agents in such allocation does not change the cost for any agent.

    Finally, we show that all item agents corresponding to the same item~$s_i$ are allocated to the same star~$C_j$.

    \begin{claim}
        For every~$i\in[3N]$, all item agents~$a_i^1,\ldots,a_i^{s_i}$ are allocated to leaves of the same star~$C_j$.
    \end{claim}
    \begin{claimproof}
        If an item agent is allocated to a center, then there is special agent not allocated to a center, which contradicts the previous claim. Hence, all item agents are allocated to leaves.
        For the sake of contradiction, assume that it is not the case, that is, there are at least two stars~$C_j$ and~$C_{j'}$ such that at least one agent of~$a_i^\ell$,~$i\in[3N]$ and~$\ell\in[s_i]$, is allocated to them. Let~$C_j$ be a star with the largest number of item agents corresponding to item~$s_i$, and let~$a_i^\ell$ be agent such that~$\alloc(a_i^\ell) \in V(C_{j'})$. Since~$a_i^\ell$ has at most by one less agents he cares about in~$C_{j'}$ than in~$C_j$, by swapping position with some agent not in~$\relSet_{a_i^\ell}$ (which clearly exists since~$|\relSet_{a_i^\ell}| < B/2$ and~$|V(C_j)| = B+1$), he decreases his cost by at least one.
    \end{claimproof}

    Now, from the fact that~$g$ is allocated to~$v_g$, that all item agents corresponding to the same~$s_i$ are part of the same star gadget, that each star has exactly~$B$ leaves, and that for every~$s_i$ we have more than~$B/4$ item agents, we directly obtain that if we take~$\mathcal{X} = (X_1,\ldots,X_N)$ so that~$X_j = \{ s_i \mid \alloc(a_i^1) \in V(C_j) \}$, we obtain a solution for~$S$.
\end{proof}

Arguably, one of the structurally simplest graphs is a path. Surprisingly, we show that even on such trivial topologies, our problem remains \NPc.

\begin{theorem}\label{thm:EF:NPh:path}
    It is \NPc to decide whether a distance preservation game~$\Gamma$ admits an envy-free allocation, even if the topology is a path.
\end{theorem}
\begin{proof}
    Again, we reduce from the \probName{$3$-Partition} problem. Let~$S$ be a \probName{$3$-Partition} instance. We construct an equivalent distance preservation game~$\Gamma$ as follows. The topology~$G$ is a path with~$N\cdot B + N + 1$ vertices. The set of agents~$\agents$ contains~$N+1$ \emph{boundary agents}~$b_1,\ldots,b_{N+1}$ and~$s_i$-many \emph{element agents}~$a_i^1,\ldots,a_i^{s_i}$ for every~$s_i\in S$. Every boundary agent~$b_i$,~$i\in[2,N+1]$, requires the boundary agent~$b_{i-1}$ at distance~$B+1$ and~$b_1$ requires~$b_2$ at distance~$B+1$. Similarly, for every~$i\in[3N]$, each agent of the element~$a_i^j$,~$j\in[2,s_i]$, must be at distance one from~$a_i^{j-1}$, and~$a_i^1$ must be at distance one from~$a_i^2$. Observe that the number of agents and the number of vertices of the graph are the same, so whenever an agent's preferences are not met by some allocation~$\alloc$, this allocation is not envy-free.

    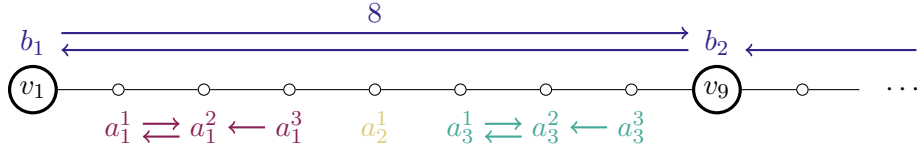
\begin{figure}[bt!]
        \centering
        \begin{tikzpicture}[scale=1.5]
            \node[draw,circle,inner sep=2pt,very thick,label=90:{\textcolor{cbBlue}{$b_1$}}] (B1) at (-1.25,0) {$v_{1}$};

            \node[draw,circle,inner sep=2pt,very thick,label=90:{\textcolor{cbBlue}{$b_2$}}] (B2) at (4.75,0) {$v_{9}$};
            
            \foreach[count=\i] \x in {-0.5,0.25,...,4} {
                \node[draw,circle,inner sep=1.5pt] (v\i) at (\x,0) {};
            }
            \node[draw,circle,inner sep=1.5pt] (v10) at (5.5,0) {};
            \draw (B1) -- (v1) -- (v2) -- (v3) -- (v4) -- (v5) -- (v6) -- (v7) -- (B2) -- (v10) -- (6,0);
            \node at (6.4,0) {$\cdots$};

            \node[below of=v1,cbRed,node distance=0.5cm] (a11) {$a_1^1$};
            \node[below of=v2,cbRed,node distance=0.5cm] (a12) {$a_1^2$};
            \node[below of=v3,cbRed,node distance=0.5cm] (a13) {$a_1^3$};

            \node[below of=v4,cbOrange,node distance=0.5cm] (a21) {$a_2^1$};

            \node[below of=v5,cbGreen,node distance=0.5cm] (a31) {$a_3^1$};
            \node[below of=v6,cbGreen,node distance=0.5cm] (a32) {$a_3^2$};
            \node[below of=v7,cbGreen,node distance=0.5cm] (a33) {$a_3^3$};

            \draw[thick,->,cbBlue] (-1,0.5) -- (4.5,0.5) node[midway,above] {$8$};
            \draw[thick,<-,cbBlue] (-1,0.35) -- (4.5,0.35);
            \draw[thick,<-,cbBlue] (5,0.35) -- (6.5,0.35);

            \draw[thick,->,cbRed] ([yshift=-0.2cm]a11.north east) -- ([yshift=-0.2cm]a12.north west);
            \draw[thick,<-,cbRed] ([yshift=0.15cm]a11.south east) -- ([yshift=0.15cm]a12.south west);
            \draw[thick,<-,cbRed] (a12) -- (a13);

            \draw[thick,->,cbGreen] ([yshift=-0.2cm]a31.north east) -- ([yshift=-0.2cm]a32.north west);
            \draw[thick,<-,cbGreen] ([yshift=0.15cm]a31.south east) -- ([yshift=0.15cm]a32.south west);
            \draw[thick,<-,cbGreen] (a32) -- (a33);
        \end{tikzpicture}
        \caption{An illustration of the construction used to prove \Cref{thm:EF:NPh:path}. An arrow from an agent~$a$ to an agent~$b$ represents that~$b\in M_a$.}
        \label{fig:EF:NPh:path}
    \end{figure}

    For correctness, let~$S$ be a \Yes-instance and~$\mathcal{X}=(X_1,\ldots,X_N)$ be a~$3$-partition of~$S$. We allocate~$b_1$ to~$v_1$, and every~$b_i$,~$i\in[2,N+1]$, to~$v_{(i-1)\cdot B + i}$. After all boundary agents are allocated, there are exactly~$B$ empty vertices between~$b_i$ and~$b_{i+1}$ for every~$i\in[N]$. Next, for every~$j\in[N]$ and every~$X_j = \{s_{i_1},s_{i_2},s_{i_3}\}$, we allocate all element agents corresponding to~$s_i$,~$i\in\{i_1,i_2,i_3\}$, on consecutive vertices between~$b_j$ and~$b_{j+1}$. Note that since~$\mathcal{X}$ is a solution, we have~$\sum_{s\in X_j} s = B$, which is exactly the number of empty vertices between~$b_j$ and~$b_{j+1}$, and also the number of agents corresponding to~$s_{i_1}$,~$s_{i_2}$ and~$s_{i_3}$. If, for each~$s_i$, we increasingly order the corresponding agents according to their identifications, we easily find that the cost for every agent is zero. Therefore,~$\alloc$ is envy-free.

    In the opposite direction, let~$\alloc$ be an envy-free allocation for~$\Gamma$. Recall that there are as many agents as there are vertices, so every vertex is occupied. First, we show that the distance between each pair of consecutive boundary agents is exactly one. For the sake of contradiction, let it not be the case; that is, there are~$b_i$ and~$b_{i+1}$,~$i\in[N]$, such that~$\dist_G(\alloc(b_i),\alloc(b_{i+1}))$. Since the topology~$G$ is a path, for every~$v\in V(G)$, there always exists at least one vertex~$u$ at the distance~$t+1$. Moreover, all vertices are occupied, so the boundary agents would envy a location at the right distance. Consequently, without loss of generality, we can assume that~$\alloc(b_i) = v_{(i-1)\cdot B + i}$. That is, the boundary agent partitions the path into~$N$ sub-paths~$P_1,\ldots,P_N$ of length exactly~$B$. Now, by the same argument as in the case of boundary agents, the set of edge agents must be allocated to a consecutive subset of vertices of some~$P_j$. We define~$\mathcal{X}=(X_1,\ldots,X_N)$ so that we set~$X_j = \{ s_i \mid a_i^1 \in V(P_j) \}$. Since there are more than~$B/4$ and less than~$B/2$ element agents corresponding to each~$s_i$ and the size of each~$V(P_j)$ is exactly~$B$, we obtain that each~$X_j$ is of size exactly three. Moreover, as each~$v\in V(P_j)$ is occupied in~$\alloc$, it must be the case that~$\sum_{s\in X_j} s = B$, as otherwise, agents corresponding to some elements~$s_i$ are not allocated to consecutive vertices of some~$P_j'$, which is a contradiction.
\end{proof}

\subsection{Bounded Number of Agents}

So far, all our hardness proofs have heavily relied on the fact that the number of agents in the instance is equal to the number of vertices of the topology. Due to this, agents were allowed to envy their ideal location. Hence, it is natural to ask whether the hardness of our problem is preserved if some vertices remain empty in every allocation.

In the first result of this subsection, we show that it is the case. Specifically, we show that the problem remains \NPc, even if the number of agents is significantly smaller compared to the number of vertices of the topology~$G$.

\begin{theorem}
\label{thm:EF:NPh:boundedAgents}
    It is \NPc to decide whether a distance preservation game~$\Gamma$ admits an envy-free allocation, even if~${\numAgents \leq \frac{3}{5}|V(G)|+1}$.
\end{theorem}
\begin{proof}
    The proof is a modification of a textbook reduction from CNF-SAT to \probName{Clique} (see, e.g., \citep[Theorem 34.11]{CormenEtal_book}).
    We reduce from the $3$-CNF-SAT problem, which is known to be \NPc \citep{GareyJ1975}. 
    In this problem, we are given a set~$X = \{x_1,\ldots, x_n\}$ of~$n$ variables and a boolean formula~$F$ in CNF with~$m$ clauses, where each clause contains precisely~$3$ literals (either a variable or its negation).~$3$-CNF-SAT then asks whether it is possible to assign each variable either True or False such that each clause contains at least one literal set to True (note that if a literal~$x$ is set to True, then its negation~$\neg x$ is set to False and vice versa). 
    
    Now, given an instance of~$3$-CNF-SAT with the set of variables~$X$ and the set of clauses~$\mathcal{C}$, we construct an instance~$\Gamma$ of DPG as follows. 
    Let~$|A| = 3m+1$,~$\relSet_a = A\setminus \{a\}$ for all~$a\in A$ and~$\dFn_a(b) = 1$ for all~$a\in A$ and~$b\in \relSet_a$. Finally, the topology~$G$ is as follows: 
    \begin{itemize}
        \item~$V(G) = V_1\cup V_2$,~$E(G) = E_1 \cup E_2$.
        \item~$V_1 = \{v_{\ell}^C\mid \ell \text{ is a literal in } C\in \mathcal{C} \}$; that is there are exactly~$3m$ many vertices in~$V_1$,  a vertex for each literal in each clause.
        \item~$E_1 = \{v_{\ell_1}^{C_1}v_{\ell_2}^{C_2}\mid C_1\neq C_2\text{ and }\ell_1\neq\neg\ell_2\}$; that is there is an edge between two vertices, if they represent literals~$\ell_1$ and~$\ell_2$ in different clause such that~$\ell_1$ is not a negation of~$\ell_2$ (and vice versa). 
        \item~$|V_2| = 2m+1$ and~$V_2\cap V_1 = \emptyset$; that is~$V_2$ contains~$2m+1$ fresh vertices not in~$V_1$. 
        \item~$E_2 = \{vv_2\mid v\in V_1\cup V_2\text{ and }v_2\in V_2\setminus \{v\}\}$; that is~$E_2$ contains all edges between any vertex in~$V_2$ and any other vertex.
    \end{itemize}
    
    Notice that~$|V_1| = 3m  < |A|$, so at least one agent will be allocated to~$V_2$ in any allocation~$\alloc$. Let~$a\in A$ be allocated to vertex~$v\in V_2$ in an allocation~$\alloc$. Since~$v$ is adjacent to all other nodes, it follows that~$\cost(a,\alloc) = 0$ and also for any~$b\in A\setminus\{a\}$,~$\cost(b,\alloc^{a\leftrightarrow b}) = 0$. It follows that the occupied vertices have to form a clique, otherwise an agent~$b$ that is not allocated to a vertex adjacent to all remaining occupied vertices would envy~$a$. Since~$|V_2| = 2m+1$, and vertices in~$V_2$ are adjacent to all other nodes, there is a clique of size~$3m+1$ in~$G$ if and only if there is a clique of size~$m$ in the subgraph~$G[V_1]$ of~$G$ induced by vertices in~$V_1$. We note that~$(G[V_1],m)$ is precisely the hard instance of \probName{Clique} obtained in the textbook reduction mentioned earlier (exactly as it appears in \citep[Theorem 34.11]{CormenEtal_book}) and it follows rather straightforwardly from the following two observation: (1) each clause corresponds to an independent set of size~$3$, so we have to pick exactly one vertex, representing a literal, from each clause 
    and (2) if we ever pick a vertex representing literal~$x\in X$, then it does not have an edge to any vertex representing the negation of~$x$. Therefore, there is a clique of size~$m$ in~$G[V_1]$ if and only if we can pick in every clause a literal such that we never pick~$x$ and~$\neg x$ at the same time, or in other words (considering we assign the picked literals truth value True), if the original~$3$-CNF-SAT instance is satisfiable. 
\end{proof}

On a more positive note, if we parameterize by the number of agents, the problem is in complexity class \XP. That is, whenever the number of agents is a fixed constant, the problem is polynomial-time solvable.

\begin{theorem}\label{thm:EF:XP:numAgents}
    There is an algorithm running in~$|V(G)|^\Oh{\numAgents}$ time that decides whether a distance preservation game~$\Gamma$ admits an envy-free allocation. That is, the problem is in \XP when parameterized by the number of agents~$\numAgents$.
\end{theorem}
\begin{proof}
    The algorithm is a simple brute-force over all~$\numAgents$-sized subsets of vertices. Specifically, we try all possible vectors~$(v_1,\ldots,v_\numAgents)$, where~$v\in V(G)$, and for every such vector, we check that a) all vertices are different and b)~$\alloc$ allocating each~$a_i\in\agents$ to the vertex~$v_i$ is envy-free. As both conditions can be checked in polynomial time and there are~$|V(G)|^\Oh{\numAgents}$ possible vectors, the running time of the algorithm is~$|V(G)|^\Oh{\numAgents}$. Such an algorithm is trivially correct as we try all possible allocations.
\end{proof}

Even with considerable effort, we were not able to complement the previous algorithm with a matching complexity lower-bound. Hence, our work leaves an important open question whether, under this parameterization, the previous algorithm is optimal or is an fixed-parameter algorithm possible.

\begin{quote}\vspace{-1cm}
\begin{openProblem}
\label{prob:open:FPTvsWh:agents}
    When parameterized by the number of agents~$\numAgents$, is it fixed-parameter tractable or \Wh to decide whether a distance preservation game~$\Gamma$ admits an envy-free allocation?
\end{openProblem}
\end{quote}

We conjecture that the correct answer to the previous question is negative, i.e. the problem is \Wh with respect to the number of agents~$\numAgents$.

Since we are not able to affirmatively resolve parameterization by the number of agents, in the remainder of this subsection, we combine parameterization by the number of agents with different restrictions of the topology.

We start with the \emph{vertex cover number}, which has been successfully used in the design of efficient algorithms in various areas of artificial intelligence research; see, e.g.,~\citep{KorhonenP2015,DeligkasEGHO2021,BlazejGKPSS2023,DeligkasEIKS2025} for a few examples. The algorithm is based on the so-called \emph{kernelization} technique, which means that we, in polynomial time, preprocess the instance of our problem by dropping non-essential parts of the input and obtain the so-called \emph{kernel}---a smaller instance, which is equivalent to the original instance and whose size is upper-bounded by some function of our parameter. Thanks to its bounded size, we can brute-force all possible solutions in \FPT time.

\begin{theorem}
\label{thm:EF:FPT:agentsPlusVC}
    When parameterized by the number of agents~$\numAgents$ and the vertex cover number~$\operatorname{vc}(G)$, there is an \FPT algorithm deciding whether a distance preservation game~$\Gamma$ admits an envy-free allocation.
\end{theorem}
\begin{proof}
    Let~$\Gamma$ be a distance preservation game, and~$C$ be a vertex cover of the topology~$G$ of optimal size~${\operatorname{vc}(G) = \vartheta}$. Without loss of generality, we can assume that~$C$ is a part of the input, as otherwise, we can compute it in~$1.25284^\Oh{\vartheta}\cdot n^\Oh{1}$ time by a known algorithm~\citep{HarrisN2024}.

    As a first step, we partition the vertices of~$I=V(G)\setminus C$ into~$2^\vartheta$ (possibly empty) types according to their neighborhood in~$C$. Let the partitioning be~$\mathcal{T} = (T_1,\ldots,T_{2^\vartheta})$. The core of our algorithm is the following reduction rule, which we apply exhaustively.

    \begin{rrule}
        Let there exist~$j\in[2^\vartheta]$ such that~$|T_j| > \numAgents$ and let~$T_j = \{v_j^1,\ldots,v_j^{|T_j|}\}$. Set~$G' = G \setminus \{ v_j^i \mid i > \numAgents \}$ and continue with~$\Gamma' = (G',\agents,(M_a)_{a\in\agents},(d_a)_{a\in\agents})$.
    \end{rrule}
    \begin{claimproof}[Safeness.]
        Let~$\Gamma$ be the original game and~$\alloc$ be an envy-free allocation. We show that there is also an envy-free allocation for~$\Gamma'$. If no agent is allocated to a vertex~$v_j^i$ where~$i > \numAgents$, we can take the same allocation and clearly obtain an envy-free solution, as the distances between all non-removed vertices are still the same. Hence,~$\alloc$ uses at least one vertex~$v_j^i$ with~$j > \numAgents$. Let~$i' \leq \numAgents$ be a vertex such that~$\alloc^{-1}(v_j^{i'}) = \emptyset$. Note that such a vertex necessarily exists as there are only~$\numAgents$ agents and at least one of them is not allocated to vertices~$v_j^1,\ldots,v_j^\numAgents$. We take~$\alloc'$ such that~$\alloc'(a) = \alloc(a)$ for every~$a\in\agents$ such that~$\alloc(a) \not= v_j^i$ and~$\alloc'(\alloc^{-1}(v_j^i)) = v_j^{i'}$. Since~$v_j^i$ and~$v_j^{i'}$ are twin vertices, the cost for agent~$\alloc^{-1}(v_j^i)$ remains in~$\alloc'$ the same as it is in~$\alloc$. We can use this argument repeatedly, and we end up with~$\alloc'$ which uses only the vertices of~$v_j^1,\ldots,v_j^\numAgents$. Such~$\alloc'$ is then a valid allocation in~$G'$ and is envy-free by the previous argumentation.

        The opposite direction is clear, as if we have~$\alloc'$, an envy-free allocation for~$\Gamma'$, we can use the same allocation also for~$\Gamma$, as all vertices used by~$\alloc'$ are present in~$G$ and by adding twins, we cannot break envy-freeness.
    \end{claimproof}

    Clearly, the previous reduction rule can be applied in polynomial time. Once the rule cannot be applied anymore, we end up with a topology with~$|V(G)| \leq \vartheta + 2^\Oh{\vartheta}\cdot \numAgents$ vertices. That is, the size of~$G$ is bounded by our parameters, and we can brute-force over all possible allocations as we did in \Cref{thm:EF:XP:numAgents}. The running time of such approach is~$(\vartheta + 2^\Oh{\vartheta}\cdot \numAgents)^\Oh{\numAgents}$, which is clearly in \FPT.
\end{proof}

The next parameter we study is called \emph{neighborhood diversity}~\citep{Lampis2012}. It is a generalization of the vertex cover number. Intuitively, a graph~$G$ is of neighborhood diversity~$\delta$, if its vertices can be partitioned into~$\delta$ types~$T_1,\ldots,T_\delta$ such that a) each type induces an independent set or a clique and b) if there is an edge between vertex~$v\in T_i$ and~$u\in T_j$, then there is a complete bipartite graph between all vertices of~$T_i$ and~$T_j$. Neighborhood diversity is an important parameter in various areas, including kidney exchange~\citep{HerbetHLSS2024}, information diffusion~\citep{KnopSS2022}, social networks analysis~\citep{AzizGN2017,GaspersN2019,GrandiKLST2023}, or coalition formation games~\citep{FioravantesGM2025}. This time, the algorithm exploits the partitioning of vertices of the topology~$G$.

\begin{theorem}\label{thm:EF:nd:agents:FPT}
    When parameterized by the number of agents~$\numAgents$ and the neighborhood diversity~$\operatorname{nd}(G)$, there is an \FPT algorithm deciding whether a distance preservation game~$\Gamma$ admits an envy-free allocation.
\end{theorem}
\begin{proof}
    Let~$H$ be a type-graph of~$G$ with vertices~$T_1,\ldots,T_{\operatorname{nd}(G)}$. Observe that we can assume that~$H$ is given on the input, as it can be easily computed in polynomial time~\cite[Theorem 7]{Lampis2012}. Slightly abusing the notation, we use~$|T_j|$,~$j\in[\operatorname{nd}(G)]$, to denote the number of vertices of graph~$G$ of type~$T_j$.

    Our algorithm tries all possible mappings of agents to the vertices of~$H$. Note that there are~$|V(H)|^\numAgents = \operatorname{nd(G)}^\numAgents$ many of them. For each such mapping~$\mu$, we check that for every~$j\in[\operatorname{nd}(G)]$, it holds that~$|\mu^{-1}(T_j)| \leq |T_j|$. If this condition is not satisfied, then the number of agents allocated to the vertices of type~$T_j$ exceeds the number of vertices available for this type, so we reject the possibility and continue with another mapping. On the other hand, if the condition is satisfied, we simply construct an allocation~$\alloc$ such that for every~$j\in[\operatorname{nd}(G)]$, we arbitrarily allocate agents in~$\mu^{-1}(T_j)$ to vertices of type~$T_j$ (by the previous condition, there are enough of such vertices) and check the envy-freeness of such an allocation in polynomial time.

    To wrap up, the running time of the algorithm is~$\operatorname{nd(G)}^\numAgents \cdot |V(G)|^\Oh{1}$, which is clearly in \FPT for the combined parameter the neighborhood diversity of~$G$ and the number of agents. The correctness of the algorithm is immediate.
\end{proof}

Note that neighborhood diversity is also a more general parameter than the distance to clique, so the previous algorithm also yields the following, complementing \Cref{thm:EF:star_clique}.

\begin{corollary}
\label{cor:EF:distToClique:agents:FPT}
    When parameterized by the number of agents~$\numAgents$ and the distance to clique~$c$, there is an \FPT algorithm deciding whether a distance preservation game~$\Gamma$ admits an envy-free allocation.
\end{corollary}

We can further strengthen the previous results to a more general class of graphs. Specifically, we study topologies of bounded modular-width~\citep{GajarskyLO2013}. This parameter, which is small for dense graphs and is more restrictive than clique-width, comes with a convenient graph decomposition that can be used in the design of dynamic programming algorithms, similarly to standard algorithms over the nice tree decomposition related to the celebrated treewidth. To the best of our knowledge, we are the first to study this parameter in the area of algorithmic game theory and computational social choice.

\begin{theorem}
\label{thm:EF:mw:agents:FPT}
    When parameterized by the number of agents~$\numAgents$ and the modular-width~$\operatorname{mw}(G)$, there is an \FPT algorithm deciding whether a distance preservation game~$\Gamma$ admits an envy-free allocation.
\end{theorem}
\begin{proof}
    We assume that the graph~$G$ is connected.
    This implies that in the root node of the decomposition, every vertex of the associated type-graph has degree at least~$1$.
    Therefore, in any other node of the decomposition, the distance of the vertices assigned to the nodes is at most~$2$ (since in the root node we join both vertices with at least one common neighbor).

    Furthermore, we observe that the distance between any two vertices of~$G$ is bounded by~$\operatorname{mw}(G)$.
    If the two vertices belong to the same node in the root node of the decomposition, we are done.
    Otherwise, we observe that the two vertices are connected by a path that is a subgraph of the root node type-graph; thus, it is of length at most~$\operatorname{mw}(G)$.

    The algorithm will be a dynamic programming style working on the decomposition.
    Before we dive into it in full detail, we begin with intuition.
    We will store a table with boolean entries for every node~$\nu$ except for the root node of the modular decomposition of~$G$:
    \[
        D_\nu\left[X \subseteq \agents, \overline{\dist}\colon X \times X \to \{0,1,2\} \right] \,.
    \]
    The semantics is that we want to assign agents in~$X$ to this node and we further specify what is the prescribed distance between the vertices to which we allocate the agents to -- the value~$\overline{\dist}(a,b)$ (we note that this function is symmetric and~$\overline{\dist}(a,a) = 0$ holds for every~$a \in X$ and for distinct~$a,b \in X$ it holds~$\overline{\dist}(a,b) \in \{1,2\}$).
    In the node~$\nu$ with the type-graph~$T_\nu$ with vertex set~$[\eta]$, given~$(X, \overline{\dist})$, we try every possible splitting~$(X_1, \ldots, X_\eta)$ of~$X$ on the nodes in~$T_\nu$ as well as the following checking that the prescribed distances comply with the distances implied by the partition~$(X_1, \ldots, X_\eta)$.
    We then use the stored values~$T_i[X_i, \overline{\dist}|_{X_i \times X_i}]$ and if all of them are \texttt{true}, we set~$D_\nu[X, \overline{\dist}]$ to \texttt{true}, otherwise we set it to \texttt{false}.
    Then, in the root node, we try all partitions of~$\agents$ on the node of the type graph and all combinations of distance functions for each node.
    For every such combination for which the appropriate table entry is \texttt{true}, we
    \begin{enumerate}
        \item compute the cost for each agent~$a \in \agents$ and
        \item check if the resulting allocation is envy-free.
    \end{enumerate}
    Note that the above computation is done on a graph with the number of vertices bounded in~$\numAgents$ (since we only need vertices with assigned agents and possibly twice as many vertices to ``subdivide edges'' and acquire distances~2).

    \smallskip
    \noindent\textbf{Leaf node.}
    Recall that a leaf node is either a clique or an independent set.
    Furthermore, the observation from the beginning of the proof yields that the distances in an independent set are in fact all~$2$.
    Therefore we set~$D_{\nu}[X, \overline{\dist}] = \texttt{true}$ if and only if there are at least~$|X|$ vertices and
    \begin{description}
        \item[clique]~$\overline{\dist}(a,b) = 1$ for distinct~$a,b \in X$ or
        \item[indep.]~$\overline{\dist}(a,b) = 2$ for distinct~$a,b \in X$.
    \end{description}
    This step of the algorithm is correct, since we cannot have different distances and we cannot assign more agents than there are available vertices.
    Furthermore, since all the vertices are indistinguishable, it does not matter which particular vertices we have selected and thus the table is set to \texttt{true} if and only if a prescribed allocation exists.

    \smallskip
    \noindent\textbf{Internal node.}
    Given a type-graph~$T_\nu$ with vertex set~$[\eta]$ for an internal node~$\nu$,~$X \subseteq \agents$, and a symmetric function~$\overline{\dist}\colon X \times X \to \{0,1,2\}$ we proceed as follows.
    Recall that~$\eta \le \operatorname{mw}(G)$.
    We generate all partitions~$\mathcal{X} = (X_1, \ldots, X_\eta)$ of the set~$X$.
    For each such partition, we compute
    \begin{align*}
        T_\nu[X, \overline{\dist}] &:= \bigwedge_{i \in [\eta]} T_i[X_i, \overline{\dist}|_{X_i \times X_i}] \wedge
        \bigwedge_{\substack{i,j \in \binom{\eta}{2}, \\ a \in X_i, b \in X_j}} \min\{\dist(i,j),2\} = \overline{\dist}(a,b)
        \,.
    \end{align*}
    As for the correctness if there is an allocation of~$X$ to the graph represented by~$T_\nu$ complying with~$\overline{\dist}$, then there exists a partition~$\mathcal{X}$ and an allocation for~$\mathcal{X}$ that complies with~$\overline{\dist}$
    \begin{enumerate}
        \item  on the type-graph-level which we check using~$\bigwedge_{\substack{i,j \in \binom{\eta}{2}, \\ a \in X_i, b \in X_j}} \min\{\dist(i,j),2\} = \overline{\dist}(a,b)$ as well as
        \item on the intra-type-level that is checked by~$T_i[X_i, \overline{\dist}|_{X_i \times X_i}]$.
    \end{enumerate}
    The other direction follows from the same argumentation -- if for partition~$\mathcal{X}$ there is an assignment that complies with~$\overline{\dist}$ on both type-graph and intra-node-level, then this assignment complies with~$\overline{\dist}$ in the graph generated by the decomposition.

    \smallskip
    \noindent\textbf{Root node.}
    Let~$T_\nu$ be the type-graph for the root node of the modular decomposition of~$G$ and let~$[\eta]$ be its vertex set.
    We try all partitions~$\mathcal{X} = (X_1, \ldots, X_\eta)$ of the set~$\agents$ and combinations of all symmetric functions~$\overline{\dist}_i \colon X_i \times X_i \to \{0,1,2\}$ for~$i\in [\eta]$.
    Fix~$\mathcal{X} = (X_1, \ldots, X_\eta)$ and a collection~$(\overline{\dist}_i)_{i=1}^\eta$.
    Note that this allows us to compute the distance for every pair of distinct agents~$a,b \in \agents$.
    Either~$a$ and~$b$ belong to the same part~$X_i$, in which case we use~$\overline{\dist}_i(a,b)$ or~$a \in X_i$ and~$b \in X_j$ in which case we use~$\dist(i,j)$ in~$T_\nu$.﻿﻿
    Given such a~$\dist$ function, we can compute for each agent~$a \in \agents$ their cost from the definition.
    Furthermore, we can calculate for every pair of distinct agents~$a,b \in \agents$ if~$a$ envies the position of~$b$.
    Now, if for at least one trial we succeed, we return \Yes, otherwise we return \No.
    In the opposite direction, assume that there exists an envy-free allocation of~$\agents$ in~$G$.
    Then, this allocation naturally yields a partition of~$\agents$ and distance functions for the actual nodes and vertices.
    Recall that within the nodes we can have distance at most~2 and thus, there exists a selection of functions~$(\overline{\dist}_i)_{i=1}^\eta$ compliant with this solution.
    If this is the case, all associated DP tables must be set to \texttt{true}; consequently, our algorithm returns \Yes.
    If there is no envy-free allocation, then for all envy-free distance functions~$\overline{\dist}_i$ there is at least one for which we cannot verify the existence of an allocation compliant with it.
    Thus, the algorithm returns \No.

    The time complexity of the above algorithm comes from the computation in each of the nodes.
    The computation performed on the leaf nodes is clearly~$\Oh{\agents}$.
    In every internal node, as well as in the root node, we perform an exhaustive search for partition~$\mathcal{X}$ and distance functions~$\overline{\dist}$.
    There are at most~$|V(G)|$ such nodes,~$\operatorname{mw}(G)^{\numAgents}$ such partitions, and~$2^{\binom{\numAgents}{2}}$ such functions (these only assign~$\{ 1,2 \}$ and are symmetric).
    Overall, the running time can be upper-bounded by
    \begin{align*}
        |V(G)| \cdot \operatorname{mw}(G)^{\numAgents} \cdot 2^{\binom{\numAgents}{2}}
        &\le
        |V(G)| \cdot \operatorname{mw}(G)^{\numAgents} \cdot 2^{\numAgents^2} \\
        &\le
        |V(G)| \cdot 2^{\Oh{\numAgents \cdot (\log \operatorname{mw}(G) + \numAgents)}} \,,
    \end{align*}
    yielding an \FPT time with respect to~$\operatorname{mw}(G)$ and~$\numAgents$.
\end{proof}

It remains unclear how much further we can push the tractability boundary. That is, it would be interesting to provide the complexity classification of the problem with respect to, e.g., tree-width or clique-width of the topology.

We conclude this subsection with one more positive result. This time, we combine the number of agents with the diameter of the topology. Note that small-world property (and, therefore, bounded diameter) has been observed, both theoretically and empirically, in many real-world networks~\citep{WattsS1998,Walsh1999}.

The algorithm is again based on kernelization. However, this time, we use a surprising connection between our problem and a famous theorem of Ramsey, which is very well-known from combinatorics.

\begin{lemma}[Ramsey's Theorem (see, e.g., \citep{erdos1935combinatorial} or \citep{balister2024upper}]\label{lem:ramsey}
    Let~$G$ be a complete graph, whose edges are colored with~$q$ different colors, and $k\in\mathbb{N}$ be an integer. If~$|V(G)|\ge q^{qk}$, then there exists~$S\subseteq V(G)$ such that~$|S|\ge k$ and all edges between the vertices in~$S$ have the same color. 
\end{lemma}

Informally speaking, Ramsey's theorem states that for every integer~$k$, each large enough edge-colored graph~$G$ contains a monochromatic clique of size~$k$. We use this as follows. If our graph is large enough with respect to our parameters, using the previous lemma on an auxiliary graph, we find a set of vertices of size~$\numAgents$ such that the vertices are equidistant. Then, any allocation of the agents on such locations is an envy-free solution. Otherwise, the size of our graph is bounded with respect to the parameters, and we can again brute-force over all possible allocations.

\begin{theorem}
\label{thm:EF:diameter:agents:FPT}
     When parameterized by the number of agents~$\numAgents$ and the diameter of the topology, there is an \FPT algorithm deciding whether a distance preservation game~$\Gamma$ admits an envy-free allocation.
\end{theorem}
\begin{proof}
    Let~$\Gamma=(\agents,(\relSet_a)_{a\in\agents},(\dFn_a)_{a\in\agents},G)$ be the input instance and let~$d$ be the diameter of~$G$. If~$|V(G)|\le d^{d\numAgents}$, then we can try all~$\binom{d^{d\numAgents}}{\numAgents}\le d^{d\numAgents^2}$ possible subset~$S$ of vertices of size~$\numAgents$ and all~$\numAgents!$ many allocation of agents to vertices in~$S$. For each allocation, we can in polynomial time compute all pairwise distances and the cost for each agent and check whether the allocation is  envy-free. Hence, if~$|V(G)|\le d^{d\numAgents}$, we can decide whether~$\Gamma$ admits an  envy-free allocation in \FPT time.
    
    Now, note that if for an allocation~$\alloc$ the distance between any pair of occupied vertices is the same, then for any pair of agent~$a,b\in \agents$ we have~$\cost(a,\alloc) = \cost(a, \alloc^{a\leftrightarrow b})$ and~$\alloc$ is envy-free. We show that if~$|V(G)| > d^{d\numAgents}$, then \Cref{lem:ramsey} implies that~$G$ has a subset of~$\numAgents$ many vertices with the same pairwise distance and hence~$\Gamma$ admits an envy-free allocation. 
    Let us define an auxiliary edge-colored complete graph~$G^{\texttt{col}}$ with the color-set~$[d]$ as follows:~$V(G^{\texttt{col}}) = V(G)$,~$E(G^{\texttt{col}}) = \{ab\mid a,b\in V(G)\text{ and }a\neq b\}$ (that is~$G^{\texttt{col}}$ is a complete graph on the same vertex set as~$G$). We color the edge~$e = ab\in E(G^{\texttt{col}})$ by the color~$\dist_G(a,b)$. By \Cref{lem:ramsey},~$G^{\texttt{col}}$ contains a set of vertices~$S$ of size~$\numAgents$ such that all edges between pairs of vertices in~$S$ have the same color. By the definition of~$G^{\texttt{col}}$, this implies that all pairs of vertices in~$S$ are at the same distance in~$G$. But, as we already observed, it follows that if we allocate agents on~$S$ arbitrarily, we get an envy-free allocation. 
\end{proof}

\subsection{Restricted Preferences}

The last dimension of the problem we can exploit in order to obtain tractable algorithms are agent's preferences. One natural restriction could be a bound on the size of relationship sets. And indeed, if each relationship set is empty, an arbitrary allocation is a solution. However, even though it was not stated formally, the construction used in the proof of \Cref{thm:EF:NPh:path} is so that the size of each relationship set is at most one. Therefore, this direction is not promising.

Consequently, we focus on even more structured preference graphs. In particular, in our first result, we assume that~$\Gpref$ is an in-star. This means that there is one ``super-star'' agent~$a^*$ and for each agent~$b\in\agents$ we have~$\relSet_b = \{a^*\}$ or~$\relSet_b = \emptyset$. The algorithm is, relatively to the restriction, surprisingly non-trivial dynamic programming.

\begin{theorem}\label{thm:EF:Gpref:instar:poly}
    If~$\Gpref$ is an in-star, there is a polynomial time algorithm deciding whether a distance preservation game~$\Gamma$ admits an envy-free allocation
\end{theorem}
\begin{proof}
    \newcommand{\DP}{\operatorname{DP}}
    Let~$a^*$ be the only agent in which other agents are interested. First, we guess the vertex to which this agent is allocated in a hypothetical solution~$\alloc$. Let this vertex be called~$v$.

    Next, we run the Breadth-First Search (BFS) algorithm from the vertex~$v$ to partition the rest of the vertices of the topology into layers~$L_1,\ldots,L_\ell$ according to their distance from~$v$. Furthermore, we partition the remaining agents into~$\ell+1$ potentially empty parts~$\agents_0,\ldots,\agents_\ell$ such that each~$\agents_0$ contains all agents with empty relationship set,~$\agents_j$,~$j\in[\ell-1]$, contains all agents~$a\in\agents$ with~$\dFn_a(a^*) = j$, and finally~$\agents_\ell$ contains all agents with~$\dFn_a(a^*) \geq \ell$. Before we introduce the core sub-procedure of our algorithm, we show an auxiliary lemma that is crucial for its correctness.

    \begin{claim}\label{thm:EF:Gpref:instar:poly:solStructure}
        Let~$\agents_j$ be a part and~$a\in\agents_j$ be an agent of this part so that~$\alloc(a) = L_{j+x}$. Then, there is no agent~$b\in\agents$ such that~$\alloc(b) \in L_{j-|x|+1} \cup \cdots\cup L_{j+|x|-1}$.
    \end{claim}
    \begin{claimproof}
        Let~$a$ be an agent of~$\agents_j$ allocated to the farthest layer of~$L_j$, say~$L_{j+x}$ from some~$x\not= 0$, and, for the sake of contradiction, assume that there is an agent~$b\in\agents$ such that~$\alloc(b) \in L_{j+y}$, where~$|y| < |x|$. Then, we have~$\cost_a(\alloc) = |x|$. However,~$\cost_a(\alloc^{a\leftrightarrow b}) = |y|$, which means that~$\cost_a(\alloc) > \cost_a(\alloc^{a\leftrightarrow b})$. This contradicts that~$\alloc$ is an envy-free allocation.
    \end{claimproof}

    Intuitively, the claim settles that the cost of all agents of the same~$\agents_j$ is the same, and, moreover, there cannot be an agent allocated closer to~$L_j$ than any agent of~$\agents_j$. Thanks to this property of outcomes, we can find a solution, if one exists, using dynamic programming over the~$A_j$'s and, simultaneously, over the BFS layers.

    Formally, we have a dynamic programming table~$\DP[j,i,e,s]$, where 
    \begin{itemize}
        \item~$j\in[\ell]$ is the number of the currently processed non-empty part,
        \item~$i\in[\ell]$ is the number of the currently processed layer,
        \item~$e\in[|L_i|]_0$ is the number of empty vertices in the layer~$L_i$, and
        \item~$s\in[\agents_0]_0$ is the number of already allocated agents of~$A_0$.
    \end{itemize}
    We call the tuple~$\sigma=(j,i,e,s)$ a \emph{signature}. Given a signature~$(j,i,e,s)$, the dynamic table~$\DP$ stores \texttt{true} if and only if there exists a \emph{partial allocation}~$\alloc_{\sigma}$ (we drop the superscript if there is no possibility of confusion) of~$\{a_0^1,\ldots,a_0^{s}\} \cup \bigcup_{k=0}^{j} \agents_k$ to~$L_1,\ldots,L_i$ such that
    \begin{enumerate}
        \item at least one agent of~$A_j$ is allocated to~$L_i$,
        \item exactly~$e$ vertices of~$L_i$ are not used by the allocation,
        \item the partial allocation satisfies \Cref{thm:EF:Gpref:instar:poly:solStructure}.
    \end{enumerate}
    If such a partial allocation~$\alloc_\sigma$ is not possible, the table stores for~$\sigma$ the value \texttt{false}.

    Now, we describe how the dynamic table is computed. 
    We start with the base case of~$i = 1$; that is, we are in the first layer. As there is no previous layer, we have to allocate all vertices of~$A_1,\ldots,A_j$ together with~$s$ vertices of~$A_0$ to this layer. Hence, we just check whether there are enough vertices in~$L_1$ or not and store the corresponding value.
    \[
        \DP[j,1,e,s] = \begin{cases}
            \texttt{true} & \text{if } e + s + \sum_{k=1}^j |\agents_k| = |L_1| \text{ and}\\
            \texttt{false} & \text{otherwise.}\\
        \end{cases}
    \]
    
    The second base case occurs when~$j$ is equal to the smallest~$k\in[\ell]$ such that~$|A_k| \not= \emptyset$; we denote this value by~$j_{\min}$. Additionally, let~$\delta = |i - j_{\min}|$. We distinguish two cases based on whether~$i \leq j_{\min}$ or~$i > j_{\min}$. In the first case, we have to allocate all agents of~$|A_{j_{\min}}|$ to~$L_i$, while in the second case, we can split them between layers~$L_i$ and~$L_{j_{\min} - \delta}$. The computation is then defined as follows.

    \[
        \DP[j,i,e,s] = \begin{cases}
            \texttt{true}& 
            \text{if } i \leq j_{\min} \text{ and } \\
            & |A_{j_{\min}}| + e + s \geq |L_i| \text{ and } \\ & |A_{j_{\min}}| + e \leq |L_i|
            \text{ and } \\ & \sum_{k=1}^{i-1} |L_k| \geq s - (|L_i| - |A_{j_{\min}}|-e)\,, \\
            \texttt{true} & 
                \text{if } i > j_{\min} \text{ and }|A_{j_{min}}| + e + s = |L_i|\,,\\
            \texttt{true} & 
                \text{if } i > j_{\min} \text{ and } \delta >= 1
                \text{ and }\\& \exists x_1\in[|A_{j_{\min}}|]
                \exists x_2\in[|A_{j_{\min}}|]_0\\&\phantom{xxx}\text{s.t. } x_1 + x_2 = |A_{j_{\min}}| 
                \text{ and }\\&
                \exists s_1,s_2\in[s]_0\colon s_1 + s_2 = s \\&\phantom{xxx}
                    \text{s.t. } x_1 + s_1 + e = |L_i|\text{ and } \\&\phantom{xxx}
                    x_2 \leq |L_{j_{\min} - \delta}|\text{ and}\\&\phantom{xxx}
                     \sum_{k=1}^{j_{\min}-\delta}|L_k| \geq x_2 + s_2\,\text{, and} \\
            \texttt{false} & \text{otherwise.}
        \end{cases}
    \]

    \newcommand{\jp}{j_{\operatorname{prev}}}
    Now, assume that~$i \geq 2$ and~$j > j_{\min}$. To simplify the description of the computation, we distinguish two cases:~$i \leq j$ and~$i > j$. We use~$\jp$ to denote the largest~$k < j$ such that~$A_k \not= \emptyset$. Observe that such~$\jp$ always exists; at worst,~$\jp = j_{\min}$. We use~$\delta'$ to denote~$|i - \jp|$ and set~$s_i = |L_i|-|A_j|-e$, that is, if~$s_i > 0$, then it is the number of agents of~$A_0$ allocated to~$L_i$.
    
    \[
        \DP[j,i,e,s] = \begin{cases}
            \DP[\jp,i,e+|A_j|,s] &\\
                \phantom{xxxxx}\text{if } \jp \geq i \text{ and } e + |A_j| \leq |L_i|\,,&\\
            \DP[\jp,i,e+|A_j|,s] &\\
                \phantom{xxxxxxx}\lor \bigvee\limits_{\mathclap{\substack{i'\in[\jp+1,i-1]\\e'\in[|L_{i'}|]_0}}}
                \DP[\jp,i',e',\max(0,s-s_i - \sum\limits_{\mathclap{k\in[i'+1,i-1]}}|L_k|)]
            &\\
                \phantom{xxxxxxx}\lor \bigvee\limits_{\mathclap{\substack{i'\in[\max(1,\jp-\delta'),\jp]\\e'\in[|L_{i'}|]_0}}}
                \DP[\jp,i',e',\max(0,s-s_i - \sum\limits_{\mathclap{k\in[\jp-i',i-1]}}|L_k|)]
                &\\
                \phantom{xxxxx}
                \text{if } \jp < i \text{ and } e + |A_j| \leq |L_i|&\\\phantom{xxxxxxxx} \text{ and } e + s + |A_j| \geq |L_i| \text{\,, and} &\\
            \texttt{false} &\\ \phantom{xxxxx}\text{otherwise.}
        \end{cases}
    \]

    To finalize the description of the computation, let~$i > j$ and~$\delta = i - j$. Now, it may happen that, in a solution, the agents of~$\agents_j$ are divided between~$L_i = L_{j+\delta}$ and~$L_{j-\delta}$. However, due to \Cref{thm:EF:Gpref:instar:poly:solStructure}, it has to be the case that there are no agents allocated to the layers between these two.
    \[
        \DP[j,i,e,s] = \begin{cases}
            \DP[\jp,i,e+|A_j|,s] &\\
                \phantom{xxxxx}\text{if } j - \delta \leq 0 \text{ and } e + |A_j| \leq |L_i|\,,&\\
            \DP[\jp,i,e+|A_j|,s] &\\
            \phantom{xxxxxxx}\lor \bigvee\limits_{\mathclap{\substack{i'\in[\max(1,\jp-\delta'),j-\delta]\\e'\in[|L_{i'}|]_0}}} 
                \DP[\jp,i',e',\max(0,s-s_i)]&\\
            \phantom{xxxxxxx}\lor \bigvee\limits_{\mathclap{\substack{x_1,x_2\in[|A_j|]\\x_1+x_2 = |A_j|\\|L_i|-x_1-e \leq s\\e'\in[|L_{i'}|-x_2]_0}}} 
                \DP[\jp,j-\delta,e',s-(|L_i|-x_1-e)] 
                &\\
            \phantom{xxxxx}\text{if } j - \delta > 0\text{ and }\jp \geq j-\delta\,,&\\
            \DP[\jp,i,e+|A_j|,s]&\\
            \phantom{xxxxxxx}\lor \bigvee\limits_{\mathclap{\substack{i'\in[\jp,j-\delta]\\e'\in[|L_{i'}|]_0}}} 
                \DP[\jp,i',e',\max(0,s-s_i-\sum\limits_{\mathclap{k=i'+1}}^{j-\delta}|L_k|)]&\\
            \phantom{xxxxxxxxx}\lor \bigvee\limits_{\mathclap{\substack{i'\in[\max(1,\jp-(j-\jp)),\jp-1]]\\e'\in[|L_{i'}|]_0}}} 
                \DP[\jp,i',e',\max(0,s-s_i-\sum\limits_{\mathclap{k=\jp+i'}}^{j-\delta}|L_k|)]&\\
            \phantom{xxxxxxx}\lor \bigvee\limits_{\mathclap{\substack{x_1,x_2\in[|A_j|]\\x_1+x_2=|A_j|\\e'\in[|L_{i'}|-x_2]_0\\s'\in[|L_{j-\delta}|-x_2]_0}}} 
                (\DP[\jp,j-\delta,e',s-(|L_i|-e-x_1)]&\\
            \phantom{xxxxxxx}\bigvee\limits_{\mathclap{\substack{i'\in[\jp,j-\delta-1]\\e'\in[|L_{i'}|]_0}}}
                \DP[\jp,i',e',\max(0,s-(|L_i|-e-x_1)-s'-\sum\limits_{k=i'+1}^{j-\delta-1} |L_k|)]&\\
            \phantom{xxxxxxxxxxx}\bigvee\limits_{\mathclap{\substack{i'\in[\max(1,\jp-(j-\delta-\jp),\jp-1]\\e'\in[|L_{i'}|]_0}}}
                \DP[\jp,i',e',\max(0,s-(|L_i|-e-x_1)-s'-\sum\limits_{\jp+(\jp-i')}^{j-\delta-1} |L_k|)]&\\
            \phantom{xxxxx}\text{if } j - \delta > 0\text{ and }\jp < j-\delta\text{\,, and}&\\
            \texttt{false}&\\
            \phantom{xxxxx}\text{otherwise.}
        \end{cases}
    \]

    We show the correctness of the computation using two auxiliary lemmas. First, we show that whenever the table stores \texttt{true} for some signature~$\sigma$, then there exists a corresponding partial partition~$\alloc_\sigma$.

    \begin{claim}
        If~$\DP[\sigma] = \texttt{true}$ for some signature~$\sigma$, then there exists a corresponding partial allocation~$\alloc_\sigma$.
    \end{claim}
    \begin{claimproof}
        First, let~$i = 1$. In this case,~$\DP$ stores \texttt{true} if and only if~$e + s + \sum_{k=1}^{j}|A_k| = |L_1|$. In other words, there are enough empty vertices in~$L_1$ to accommodate all agents of~$\bigcup_{k=1}^j A_k$,~$s$ agents of~$A_0$, and keep~$e$ vertices empty. Hence, we take a partial allocation~$\alloc_\sigma$ so that we arbitrarily allocate all agents~$\bigcup_{k=1}^j A_k$ together with agents~$a_0^1,\ldots,a_0^s$ to~$V(L_i)$. \Cref{thm:EF:Gpref:instar:poly:solStructure} is easily satisfied since all agents are in the same layer and the number of empty vertices of~$L_i$ is exactly~$|L_i| - s - \bigcup_{k=1}^j |A_k| = e$.
    \end{claimproof}

    Next, we show that whenever there exists a partial solution~$\alloc_\sigma$ corresponding to a signature~$\sigma$, then also the dynamic programming table stores \texttt{true} for this signature.

    \begin{claim}
        Let~$\sigma=(j,i,e,s)$ be a signature. If there exists a partial solution~$\alloc_\sigma$ compatible with~$\sigma$, then~$\DP[j,i,e,s] = \texttt{true}$.
    \end{claim}
    \begin{claimproof}
        First, let~$i=1$ and~$\alloc_\sigma$ be a partial allocation corresponding to~$\sigma=(j,1,e,s)$. Since~$i=1$, there are no previous layers and hence all agents of~$\bigcup_{k=1}^j A_k$ together with~$s$ agents of~$A_0$ must be allocated to~$L_1$. Moreover, exactly~$e$ vertices of~$L_i$ are not used by the allocation. That is,~$|L_i| - e - s - |\bigcup_{k=1}^j A_j| = |L_i| - e - s - \sum_{k=1}^j |A_k| = 0$, which implies~$e + s + \sum_{k=1}^j |A_k| = |L_i|$. However, in this case, the dynamic programming table stores \texttt{true} for~$\sigma$ by the definition of the computation.
    \end{claimproof}

    The correctness of our approach then follows directly by the previous claims. Once the dynamic programming table~$\DP$ is correctly computed, we check whether an envy-free allocation is possible by finding~$i\in[\ell]$ and~$e\in[|L_i|]_0$ such that~$\DP[\ell',i,e,|A_0|] = \texttt{true}$, where~$\ell'\in[\ell]$ is the largest integer such that~$\agents_{\ell'} \not= \emptyset$. If no such cell exists, we return \No. For the running time, observe that there are~$\operatorname{diam}(G)\cdot\operatorname{diam}(G)\cdot n\cdot\numAgents \in \Oh{n^4}$ different signatures~$\sigma$. The most time-consuming computations are for signatures where~$j < i$, where it can take~$\Oh{n^3}$ time. Therefore, the overall running time is~$\Oh{n^7}$, which means that the algorithm runs in polynomial time.
\end{proof}

In the previous result, we assumed that there is exactly one agent the other agents care about. Next, we turn our attention to a similarly restricted case. In particular, we show that if there is exactly one agent which is not indifferent, then a polynomial time algorithm is possible.

\begin{theorem}\label{thm:EF:Gpref:outStar:poly}
    If~$\Gpref$ is an out-star, there is a polynomial time algorithm deciding whether a distance preservation game~$\Gamma$ admits an envy-free allocation.
\end{theorem}
\begin{proof}
    The algorithm is a simple swap dynamics that starts with an arbitrary allocation~$\alloc$ and, through a sequence of swaps improving agent's~$a$ cost, it reaches an envy-free allocation; see \Cref{algo:EF:Gpref:outStar:poly} for a formal description.

    \begin{algorithm}[tb]
    \caption{An algorithm finding an envy-free allocation if the preference graph~$\Gpref$ is an out-star.}
    \label{algo:EF:Gpref:outStar:poly}
    \begin{algorithmic}[1]
        \STATE Fix an arbitrary allocation~$\alloc$
        \WHILE{$\exists b\in\agents\setminus\{a\}$ s. t.~$\cost_a(\alloc) > \cost_a(\alloc^{a \leftrightarrow b})$}\label{algo:EF:Gpref:outStar:poly:EFcheck}
            \STATE~$c = \arg\max_{b\in\agents}(\cost_a(\alloc)-\cost_a(\alloc^{a\leftrightarrow b}))$  
            \STATE~$\alloc = \alloc^{a\leftrightarrow c}$
        \ENDWHILE
        \RETURN Envy-free allocation~$\alloc$
    \end{algorithmic}
    \end{algorithm}

    For correctness, it is easy to see that the algorithm always returns an envy-free allocation, since it terminates only if the condition on line \ref{algo:EF:Gpref:outStar:poly:EFcheck}, which checks for envy-freeness for agent~$a$, is not satisfied. For running time, after each execution of the loop in line \ref{algo:EF:Gpref:outStar:poly:EFcheck}, the cost of agent~$a$ is reduced by at least one (and all other agents are indifferent, so their cost is constant zero). Moreover, the initial cost is at most~$\operatorname{diam}(G)\cdot\numAgents \in \Oh{n^2}$, and the cost cannot be less than zero. That is, after at most a quadratic number of improving steps, the algorithm returns an envy-free allocation.
\end{proof}

We conclude with one more intractability result. Specifically, we show that even if the number of agents with a non-empty relationship set is our parameter, the problem is \Wh. This result is complemented by an \XP algorithm similar to that of \Cref{thm:EF:XP:numAgents}, but requires additional steps to ensure that the important agents do not envy any of the indifferent agents.

\begin{theorem}
\label{thm:EF:Gpref:nonZeroDegree:poly}
    When parameterized by the number of agents with non-zero degree in the preference graph~$\Gpref$, it is~\Wh and in \XP to decide whether a distance preservation game~$\Gamma$ admits an envy-free allocation.
\end{theorem}
\begin{proof}
    We reduce from the \probName{Clique}. In this problem, the input is a graph~$G=(V,E)$ and~$k \in\mathbb{N}$. The question is: does there exist a subset of vertices~$C\subseteq V$ such that~$|X|=k$ and~$X$ induces a complete subgraph (i.e.,~$C$ forms a clique)?

    Given an instance~$S$ of \probName{Clique} with a graph~$G=(V,E)$, we construct an equivalent distance preservation game~$\Gamma$ as follows. The topology~$H$ consists of the vertices~$V \cup u^*$, where~$u^*$ is a vertex adjacent to all vertices in~$V$. We refer to~$u^*$ as the universal vertex.
    We create~$k+1$ “interested" agents~$a_1, \dots,a_{k+1}$ and~$|V(H)| - k -1$ so-called “neutral" agents.
    For each neutral agent~$a$, we have~$M_a = \emptyset$, i.e., they have no preferences for any other agent. 
    For each interested agent~$a_i$, we have~$M_{a_i} = \{ a_1,\ldots,a_{k+1} \}\setminus \{a_i\}$. Each agent requires to be distance one from all other interested agents. 

    Assume we are given a solution~$S$ to \probName{Clique}, i.e., we have a clique of size~$k$ in graph~$G$, which is by construction also in graph~$H$. We allocate~$k$ of the interested agents to the vertices in clique~$C$ in graph~$H$ and the final interested agent to the universal vertex~$u^*$.
    All neutral agents can be arbitrarily allocated to the remaining vertices in~$H$. This allocation~$\pi$ is trivially envy-free as every interested agent is distance one from one another as they are on all placed on a clique of size~$k+1$, i.e.,~$C\cup u^*$.

    In the other direction, given an envy-free allocation~$\pi$ to the preservation game~$\Gamma$, we argue that there must exist a clique of size~$k$ in graph~$G$.
    Since the number of agents in~$\Gamma$ is the same as the number of vertices in~$V$, there must be an agent at the universal vertex~$u^*$. Recall that the neutral agents have no preference for where they are allocated. Any interested agent allocated to the universal vertex will have cost~$0$ as they are at distance~$1$ from every other vertex. Consider now some interested agent~$a_i$;~$a_i$ will envy the agent at~$u^*$ unless~$a_i$ also has cost~$0$ at his current vertex. For an agent to have cost~$0$ he must be at distance~$1$ from all of the other~$k$ interested agents. Therefore, we can only have an envy-free solution if all of the~$k$ agents are on a clique in graph~$H$ and the additional interested agent is allocated to~$u^*$. Hence, if we have an envy-free solution for the topology~$H$, we have a clique of size~$k$ in graph~$G$.

    It remains to show that deciding whether a distance preservation game~$\Gamma$ admits an envy-free allocation is in XP parameterized by the number of agents with non-zero degree in~$\Gpref$. The algorithm is a rather simple generalization of the algorithm in \Cref{thm:EF:XP:numAgents}. Let~$\agents_{>0} = \{a_1,a_2,\ldots, a_{|\agents_{>0}|}\}$ be the set of agents with non-zero degree in~$\Gpref$. We again try all vectors~$(v_1,\ldots,v_{|\agents_{>0}|})$, where~$v_i\in V(G)$, for all~$i\in [|\agents_{>0}|]$.  We check that all vertices in the vector are distinct and if we allocate agent~$a_i$ to the vertex~$v_i$ for all~$i\in [|\agents_{>0}|]$ we get an allocation that is envy-free among agents in~$\agents_{>0}$. If that is the case, we still need to allocate agents in~$\agents\setminus\agents_{>0}$. Notice that such agents have cost~$0$ wherever they are allocated and their allocation does not change the costs of agents in~$\agents_{>0}$. Moreover, while an agent~$a\in\agents_{>0}$ can still envy an agent~$b\in \agents\setminus\agents_{>0}$, the cost of~$a$ in~$\alloc^{a\leftrightarrow b}$ depends only on allocations of agents in~$\agents_{>0}$, for whom we are trying all possibilities and for whom we have fixed allocation based on the vector~$(v_1,\ldots, v_{|\agents_{>0}|})$. Hence for every vertex~$v\in V(G)\setminus \{v_1,\ldots, v_{|\agents_{>0}|}\}$ and every agent~$a_i\in \agents_{>0}$, we can compute whether the agent~$a_i$ would envy an agent~$b\in \agents\setminus\agents_{>0}$ if~$b$ were to be allocated on~$v$. This is easily done by computing the current cost of~$a_i$ and the cost after moving~$a_i$ to~$v$. This partitions vertices in~$V(G)\setminus \{v_1, \ldots, v_{|\agents{_>0}|}\}$ into two sets~$V_{\text{envy}}$ and~$V_{\text{no-envy}}$, such that if any agent in~$\agents_{>0}$ would envy an agent on vertex~$v$, then~$v$ is in~$V_{\text{envy}}$ and~$v$ is in~$V_{\text{no-envy}}$. It is now straightforward to see that if~$|V_{\text{no-envy}}| \ge |\agents\setminus\agents_{>0}|$, then we can allocate all agents with degree zero in~$\Gpref$ to an arbitrary vertex in~$|V_{\text{no-envy}}|$ and the final allocation is envy-free. Otherwise, in any extension of the allocations that allocate~$a_i$ to~$v_i$ for all~$i\in [|\agents_{>0}|]$ would allocate at least one agent~$b\in \agents\setminus\agents_{>0}$ on a vertex in~$V_{\text{envy}}$ and there would be an agent in~$\agents_{>0}$ that envies~$b$. Notice that we try all possible allocations of agents with non-zero degree in~$\Gpref$. For each such allocation, we either find an allocation of remaining agents that is envy-free or show that it is impossible to extend to an envy-free allocation. Hence, the algorithm correctly decides~$\Gamma$. Finally, we try at most~$|V(G)|^{|\agents_{>0}|}$ allocations and for each, we only need to compute the cost of an agent~$|\agents_{>0}|\cdot |V(G)|$ many times, hence the algorithm has \XP running time. 
\end{proof}

Note that parameterization by the number of indifferent agents is hopeless, as we have \NPcness already for instances with no such agent (cf. \Cref{thm:EF:NPh}).

\section{Jump and Swap Stability}

In the remainder of the paper, we briefly explore the stability notions of jump and swap stability. First, we show that these notions are also not always satisfiable. We start with jump stability and show that already with two agents, stable allocation is not guaranteed to exist.

\begin{proposition}
\label{prop:jumpSwap:twoAgents:noExist}
    A jump stable allocation is not guaranteed to exist, even if~$\numAgents = 2$ and the topology is a path.
\end{proposition}
\begin{proof}
    Assume that the topology is a path on at least four vertices, we have two agents~$a_1$ and~$a_2$ interested in each other, and~$\dFn_{a_1}(a_2) = 1$ and~$\dFn_{a_2}(a_1) = 2$. For the sake of contradiction, assume that there is an envy-free allocation. If~$a_2$ is at distance one from~$a_1$, then~$a_2$ prefers to jump to a location at distance two from~$\alloc(a_1)$, and such a vertex always exists, as the path is of length at least four. Similarly, if~$a_1$ is not at the distance exactly one from~$a_2$,~$a_1$ prefers to jump to a location at distance one from~$a_2$, which again always exists and is empty. These are the only possible allocations, so the proposition holds.
\end{proof}

For swap stability, instances with two agents always admit a swap-stable solution, as such instances always admit an envy-free solution. Hence, we need at least three agents to show non-existence.

\begin{proposition}
\label{prop:swap:notExist}
    A swap stable allocation is not guaranteed to exist, even if~$\numAgents = 3$ and the topology is a path.
\end{proposition}
\begin{proof}
    Let~$\agents = \{a_1,a_2,a_3\}$ and the topology be a path with three vertices~$v_1$,~$v_2$,~$v_3$. For every agent~$a\in\agents$, we set~$\relSet_a = \agents\setminus\{a\}$. The distance functions are then as follows:~$\dFn_{a_1}(a_2) = 1$,~$\dFn_{a_1}(a_3) = 2$,~$\dFn_{a_2}(a_1) = 2$,~$\dFn_{a_2}(a_3) = 1$,~$\dFn_{a_3}(a_1) = 1$, and~$\dFn_{a_3}(a_2) = 2$. That is, each agent wants one agent at distance one and the other at distance two, but these relations are not symmetric. For the sake of contradiction, let~$\alloc$ be a swap stable allocation. 
    Without loss of generality suppose that~$\alloc(a_1) = v_1$,~$\alloc(a_2) = v_2$, and~$\alloc(a_3) = v_3$; the other cases are symmetric to this one. Now, agent~$a_2$ is at distance one from both~$a_1$ and~$a_3$, so his cost is~$1$. Agent~$a_3$ is at distance one from~$a_2$ and at distance two from~$a_1$, but wants to have it the other way around; i.e., the cost for her is two. If these two agents swap, both of their costs decrease, so~$\alloc$ is not swap stable.
\end{proof}

If the assumed stability notion is envy-freeness, it is \NPc to decide the existence even if the preferences are symmetric. As we show in our first algorithmic result, the situation with jump and swap stability is much more positive. Specifically, by a potential argument, a simple best-response dynamic converges in polynomial time.

\begin{theorem}
\label{thm:swapJump:symmetric:polytime}
    If the preferences are symmetric, a jump stable and a swap stable allocation always exist and can be found in polynomial time.
\end{theorem}
\begin{proof}
    We prove the theorem by showing the existence of a polynomially bounded potential, and so a simple best-response dynamic converges in polynomial time.
    For an allocation~$\alloc$, we let~$P(\alloc) = \sum_{a\in \agents}(\cost(a,\alloc) - \sum_{b\in\relSet_a}\max\{0, \dFn_a(b)-|V(G)|\})$ be the sum of costs of all agents, where we normalize all distance functions to be at most~$|V(G)|$ by letting~$\dFn_a(b) = |V(G)|$ whenever~$\dFn_a(b) > |V(G)|$\footnote{Note that indeed if~$\dFn_a(b) > |V(G)|$, then because~$\dist(\alloc(a), \alloc(b))\le |V(G)|$, we have~$|\dFn_a(b)-\dist(\alloc(a), \alloc(b))| - (\dFn_a(b)-|V(G)|) = \dFn_a(b)-\dist(\alloc(a), \alloc(b)) - (\dFn_a(b)-|V(G)|) = ||V(G)| - \dist(\alloc(a), \alloc(b))|$.}. It follows that~$0\le P(\alloc)\le \numAgents^2\cdot |V(G)|$. Now, let us compute how the value of~$P(\alloc)$ changes in the case of a single jump or a single swap. First, note that~$\sum_{a\in \agents}\sum_{b\in \relSet_a}\max\{0, \dFn_a(b)-|V(G)|\}$ does not depend on~$\alloc$, so we only need to consider how the sum of costs changes. Let us first consider the jump of an agent~$a$ to a vertex~$v$. Since only~$a$ and the agents~$b$ with~$a\in \relSet_b$ are affected and since~$a\in \relSet_b$ if and only if~$b\in \relSet_a$ in a symmetric instance, we get~$P(\alloc)- P(\alloc^{a\mapsto v}) = \cost(a,\alloc) - \cost(a,\alloc^{a\mapsto v}) + \sum_{b\in M_a}(|\dFn_b(a) - \dist(\alloc(a), \alloc(b))| - |\dFn_b(a) - \dist(\alloc^{a\mapsto v}(a), \alloc^{a\mapsto v}(b))|)$. Note that~$\dFn_b(a) = \dFn_a(b)$, because the preferences are symmetric. So~$\sum_{b\in \relSet_a}|\dFn_b(a) - \dist(\alloc(a), \alloc(b))| = |\dFn_a(b) - \dist(\alloc(a), \alloc(b))| = \cost (a, \alloc)$ and similarly,~$\sum_{b\in \relSet_a}|\dFn_b(a) - \dist(\alloc^{a\mapsto v}(a), \alloc^{a\mapsto v}(b))| = \cost (a, \alloc^{a\mapsto v})$. Therefore,~$P(\alloc)- P(\alloc^{a\mapsto v}) = 2(\cost(a,\alloc) - \cost(a,\alloc^{a\mapsto v}))$. It follows that indeed~$P$ is a polynomially bounded potential. Hence, we start from an arbitrary allocation~$\alloc$ and follow a simple best response dynamic, where 1) we check for every agent whether they can improve their cost by jumping to an empty vertex, 2) if so, let the agent jump there and repeat. By the above argument, in every jump, the value of~$P(\alloc)$ decreases by two times the improvement of the cost of the agent that jumped. Hence, after at most~$\numAgents^2|V(G)|$ many jumps, we reach a jump stable allocation. By an analogous argument, we get that if~$a$ and~$b$ swap, then~$P(\alloc) - P(\alloc^{a\leftrightarrow b}) = 2(\cost(a,\alloc)-\cost(a,\alloc^{a\leftrightarrow b})+\cost(b,\alloc)-\cost(b,\alloc^{a\leftrightarrow b}))$ and starting from arbitrary allocation and letting agents swap if they both prefer the swap converges to a swap stable allocation in polynomially-many swaps. 
\end{proof}

Our last result settles the polynomial-time solvability of the problem for jump and swap stability in the case where the preferences are acyclic. The core idea of the algorithm is to find a topological ordering according to agents' preferences and let them pick their best location in the opposite of this ordering. The proof is similar to that of \citet{AzizCLNW2025} for an analogous setting in the original distance preservation games.

\begin{theorem}
\label{thm:swapJump:acyclic:polytime}
    If the preference graph~$\Gpref$ is acyclic, a jump stable and a swap stable allocation always exist and can be found in polynomial time.
\end{theorem}
\begin{proof}
    The proof is basically identical to the analogous result by \citet{AzizCLNW2025} for the model, where agents are allocated on a line segment. Since~$\Gpref$ is acyclic, we can compute a topological ordering of the vertices of~$\Gpref$ such that all edges are directed from left to right in this order. We then repeatedly select the rightmost unallocated agent, compute their cost on all empty vertices so far, and place them on a vertex that minimizes the cost. Note that since all edges are directed from left to right, when we allocate an agent~$a$, all agents in~$\relSet_a$ are already allocated. Therefore, the cost on any node for~$a$ is already fixed when we are allocating~$a$. It follows that, after the full allocation is finalized, the agent~$a$ does not want to move to any other vertex that is either empty or occupied by an agent to the left of~$a$ in the ordering that selected their vertex after~$a$. Therefore, the allocation is jump stable and swap stable (since in any pair of agents, the one to the right in the ordering does not want to swap). 
\end{proof}

\section{Discussion}\label{sec:discussion}

Our paper provides an almost complete landscape of the complexity of envy-free allocations in graphical distance preservation games. Our results illustrate that restrictions in the input instances are required to reach tractability.
We believe that Open Problem~\ref{prob:open:FPTvsWh:agents} is the most important question that remains open on this front, which deserves further study. Below, we highlight a few more future directions that merit further study.

\paragraph{``At most'' and ``at least'' distances.}
Two very natural classes of ideal distances are what we term {\em ``at least''} and {\em ``at most''} distances. 
These are two-step functions where one of the steps is zero, and the other one is linear. %
There, if agent~$i$ wants their distance from~$j$ to be at least (resp. at most)~$d$, then the cost is zero if under the allocation the condition is satisfied and defined as normal otherwise.
Observe that when an agent has ideal distance at most 1, then it is equivalent to having ideal distance exactly 1, hence the \NPcness for envy-free allocations from \Cref{thm:EF:NPh} holds for this case. 
However, the problem remains open when agents have ``at least'' ideal distances, which is arguably a very natural class. Does our problem become easier under these distance functions?

\paragraph{Jump and Swap Stability.}
Our initial results for jump and stable allocations leave ample space for further study. Of course, there are the directions of studying the complexity of the problem, for which we conjecture it is \NPc in general, under the same lens as envy-free allocations. Study how it behaves when we constrain the topology, or the graph, or agent's preferences. 
One less obvious direction that we would like to highlight is the existence of jump and swap stable allocations on paths under ``at least'' and ``at most'' distance functions. 
Despite our efforts to resolve the question, in both directions, the problem remained hard to crack. 

\paragraph{Maximize Social Welfare.}
An orthogonal research direction is to study the maximization of social-welfare in graphical distance games; this objective was extensively studied in the model of~\citet{AzizLSV2025}. Observe that this is a very interesting problem, since it is at least as hard as \probName{Subgraph Isomorphism}, an important problem for various AI subfields~\citep{HoffmannMR2017,McCreeshPST2018}; this is when~$\Gpref$ is symmetric and the ideal distance for every agent is~$1$.
What are the graph classes that allow for efficient algorithms for this problem?

\clearpage
\section*{Acknowledgments}
This project has received funding from the European Research Council (ERC) under the European Union’s Horizon 2020 research and innovation programme (grant agreement No 101002854) and was co-funded by the European Union under the project Robotics and advanced industrial production (reg. no. CZ.02.01.01/00/22\_008/0004590). Argyrios Deligkas acknowledges the support of the EPSRC grant EP/X039862/1.
    
\begin{center}
    \vspace{0.25cm}
    \includegraphics[width=5cm]{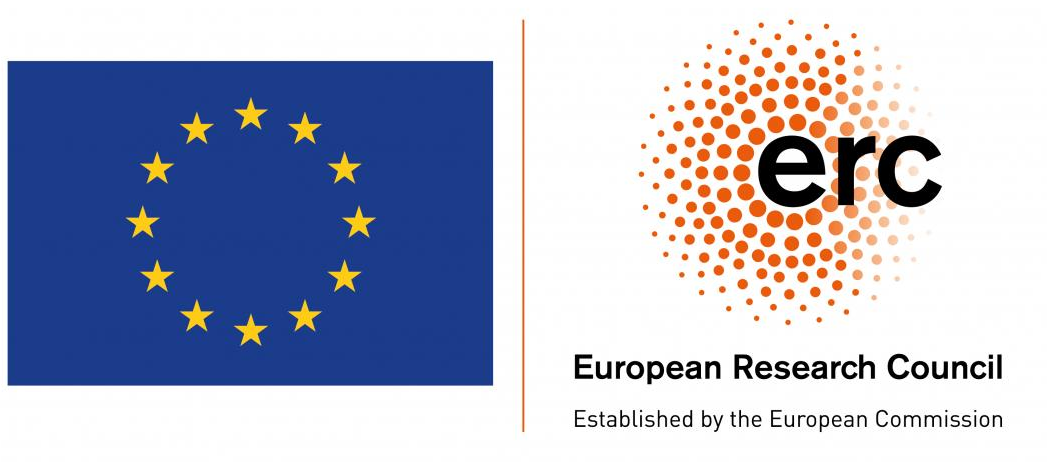}
\end{center}

\bibliographystyle{plainnat}
\bibliography{references}

\begin{thebibliography}{54}
\providecommand{\natexlab}[1]{#1}
\providecommand{\url}[1]{\texttt{#1}}
\expandafter\ifx\csname urlstyle\endcsname\relax
  \providecommand{\doi}[1]{doi: #1}\else
  \providecommand{\doi}{doi: \begingroup \urlstyle{rm}\Url}\fi

\bibitem[Agarwal et~al.(2021)Agarwal, Elkind, Gan, Igarashi, Suksompong, and
  Voudouris]{AgarwalEGISV2021}
Aishwarya Agarwal, Edith Elkind, Jiarui Gan, Ayumi Igarashi, Warut Suksompong,
  and Alexandros~A. Voudouris.
\newblock {S}chelling games on graphs.
\newblock \emph{Artificial Intelligence}, 301:\penalty0 103576, 2021.
\newblock \doi{10.1016/j.artint.2021.103576}.

\bibitem[Aziz et~al.(2017)Aziz, Gaspers, and Najeebullah]{AzizGN2017}
Haris Aziz, Serge Gaspers, and Kamran Najeebullah.
\newblock Weakening covert networks by minimizing inverse geodesic length.
\newblock In \emph{Proceedings of the 26th International Joint Conference on
  Artificial Intelligence, {IJCAI}~'17}, pages 779--785. ijcai.org, 2017.
\newblock \doi{10.24963/ijcai.2017/108}.

\bibitem[Aziz et~al.(2023)Aziz, Suksompong, Sun, and Walsh]{AzizSSW2023}
Haris Aziz, Warut Suksompong, Zhaohong Sun, and Toby Walsh.
\newblock Fairness concepts for indivisible items with externalities.
\newblock In \emph{Proceedings of the 37th {AAAI} Conference on Artificial
  Intelligence, {AAAI}~'23}, pages 5472--5480. {AAAI} Press, 2023.
\newblock \doi{10.1609/aaai.v37i5.25680}.

\bibitem[Aziz et~al.(2025{\natexlab{a}})Aziz, Chan, Lederer, Narang, and
  Walsh]{AzizCLNW2025}
Haris Aziz, Hau Chan, Patrick Lederer, Shivika Narang, and Toby Walsh.
\newblock Distance preservation games.
\newblock In \emph{Proceedings of the 34th International Joint Conference on
  Artificial Intelligence, {IJCAI}~'25}, pages 3735--3743. ijcai.org,
  2025{\natexlab{a}}.
\newblock \doi{10.24963/ijcai.2025/415}.

\bibitem[Aziz et~al.(2025{\natexlab{b}})Aziz, Lisowski, Suzuki, and
  Vollen]{AzizLSV2025}
Haris Aziz, Grzegorz Lisowski, Mashbat Suzuki, and Jeremy Vollen.
\newblock Neighborhood stability in assignments on graphs.
\newblock In \emph{Proceedings of the 24th International Conference on
  Autonomous Agents and Multiagent Systems, {AAMAS}~'25}, pages 2420--2422,
  2025{\natexlab{b}}.
\newblock URL \url{https://dl.acm.org/doi/10.5555/3709347.3743889}.

\bibitem[Balister et~al.(2026)Balister, Bollob{\'a}s, Campos, Griffiths,
  Hurley, Morris, Sahasrabudhe, and Tiba]{balister2024upper}
Paul Balister, B{\'e}la Bollob{\'a}s, Marcelo Campos, Simon Griffiths, Eoin
  Hurley, Robert Morris, Julian Sahasrabudhe, and Marius Tiba.
\newblock Upper bounds for multicolour {R}amsey numbers.
\newblock \emph{Journal of the American Mathematical Society}, 2026.
\newblock \doi{10.1090/jams/1069}.

\bibitem[Balliu et~al.(2019)Balliu, Flammini, Melideo, and
  Olivetti]{BalliuFMO2019}
Alkida Balliu, Michele Flammini, Giovanna Melideo, and Dennis Olivetti.
\newblock On non-cooperativeness in social distance games.
\newblock \emph{Journal of Artificial Intelligence Research}, 66:\penalty0
  625--653, 2019.
\newblock \doi{10.1613/jair.1.11808}.

\bibitem[Balliu et~al.(2022)Balliu, Flammini, Melideo, and
  Olivetti]{BalliuFMO2022}
Alkida Balliu, Michele Flammini, Giovanna Melideo, and Dennis Olivetti.
\newblock On {P}areto optimality in social distance games.
\newblock \emph{Artificial Intelligence}, 312:\penalty0 103768, 2022.
\newblock \doi{10.1016/j.artint.2022.103768}.

\bibitem[Berriaud et~al.(2023)Berriaud, Constantinescu, and
  Wattenhofer]{BerriaudCW2023}
Damien Berriaud, Andrei Constantinescu, and Roger Wattenhofer.
\newblock Stable dinner party seating arrangements.
\newblock In \emph{Proceedings of the 19th International Conference on Web and
  Internet Economics, {WINE}~'23}, volume 14413 of \emph{LNCS}, pages 3--20.
  Springer, 2023.
\newblock \doi{10.1007/978-3-031-48974-7\_1}.

\bibitem[Bil{\`{o}} et~al.(2022)Bil{\`{o}}, Bil{\`{o}}, Lenzner, and
  Molitor]{BiloBLM2022}
Davide Bil{\`{o}}, Vittorio Bil{\`{o}}, Pascal Lenzner, and Louise Molitor.
\newblock Topological influence and locality in swap {S}chelling games.
\newblock \emph{Autonomous Agents and Multi-Agent Systems}, 36\penalty0
  (2):\penalty0 47, 2022.
\newblock \doi{10.1007/s10458-022-09573-7}.

\bibitem[Bil{\`{o}} et~al.(2023)Bil{\`{o}}, Bil{\`{o}}, D{\"{o}}ring, Lenzner,
  Molitor, and Schmidt]{BiloBDLMS2023}
Davide Bil{\`{o}}, Vittorio Bil{\`{o}}, Michelle D{\"{o}}ring, Pascal Lenzner,
  Louise Molitor, and Jonas Schmidt.
\newblock {S}chelling games with continuous types.
\newblock In \emph{Proceedings of the 32nd International Joint Conference on
  Artificial Intelligence, {IJCAI}~'23}, pages 2520--2527. ijcai.org, 2023.
\newblock \doi{10.24963/ijcai.2023/280}.

\bibitem[Bla{\v{z}}ej et~al.(2023)Bla{\v{z}}ej, Ganian, Knop, Pokorn{\'{y}},
  Schierreich, and Simonov]{BlazejGKPSS2023}
V{\'{a}}clav Bla{\v{z}}ej, Robert Ganian, Du{\v{s}}an Knop, Jan Pokorn{\'{y}},
  {\v{S}}imon Schierreich, and Kirill Simonov.
\newblock The parameterized complexity of network microaggregation.
\newblock In \emph{Proceedings of the 37th {AAAI} Conference on Artificial
  Intelligence, {AAAI}~'23}, pages 6262--6270. {AAAI} Press, 2023.
\newblock \doi{10.1609/aaai.v37i5.25771}.

\bibitem[Bodlaender et~al.(2025)Bodlaender, Hanaka, Jaffke, Ono, Otachi, and
  van~der Zanden]{BodlaenderHJOOZ2020}
Hans~L. Bodlaender, Tesshu Hanaka, Lars Jaffke, Hirotaka Ono, Yota Otachi, and
  Tom~C. van~der Zanden.
\newblock Hedonic seat arrangement problems.
\newblock \emph{Autonomous Agents and Multi-Agent Systems}, 39\penalty0
  (2):\penalty0 33, 2025.
\newblock \doi{10.1007/s10458-025-09711-x}.

\bibitem[Br{\^{a}}nzei and Larson(2011)]{BranzeiL2011}
Simina Br{\^{a}}nzei and Kate Larson.
\newblock Social distance games.
\newblock In \emph{Proceedings of the 22nd International Joint Conference on
  Artificial Intelligence, {IJCAI}~'11}, pages 91--96. {IJCAI/AAAI}, 2011.
\newblock \doi{10.5591/978-1-57735-516-8/IJCAI11-027}.

\bibitem[Bullinger and Suksompong(2024)]{BullingerS2023}
Martin Bullinger and Warut Suksompong.
\newblock Topological distance games.
\newblock \emph{Theoretical Computer Science}, 981:\penalty0 114238, 2024.
\newblock \doi{10.1016/j.tcs.2023.114238}.

\bibitem[Bullinger et~al.(2021)Bullinger, Suksompong, and
  Voudouris]{BullingerSV2021}
Martin Bullinger, Warut Suksompong, and Alexandros~A. Voudouris.
\newblock Welfare guarantees in {S}chelling segregation.
\newblock \emph{Journal of Artificial Intelligence Research}, 71:\penalty0
  143--174, 2021.
\newblock \doi{10.1613/jair.1.12771}.

\bibitem[Ceylan et~al.(2023)Ceylan, Chen, and Roy]{CeylanCR2023}
Esra Ceylan, Jiehua Chen, and Sanjukta Roy.
\newblock Optimal seat arrangement: {W}hat are the hard and easy cases?
\newblock In \emph{Proceedings of the 32nd International Joint Conference on
  Artificial Intelligence, {IJCAI}~'23}, pages 2563--2571. ijcai.org, 2023.
\newblock \doi{10.24963/ijcai.2023/285}.

\bibitem[Chauhan et~al.(2018)Chauhan, Lenzner, and Molitor]{ChauhanLM2018}
Ankit Chauhan, Pascal Lenzner, and Louise Molitor.
\newblock {S}chelling segregation with strategic agents.
\newblock In \emph{Proceedings of the 11th International Symposium on
  Algorithmic Game Theory, {SAGT}~'18}, volume 11059 of \emph{LNCS}, pages
  137--149. Springer, 2018.
\newblock \doi{10.1007/978-3-319-99660-8\_13}.

\bibitem[Cormen et~al.(2009)Cormen, Leiserson, Rivest, and
  Stein]{CormenEtal_book}
Thomas~H. Cormen, Charles~E. Leiserson, Ronald~L. Rivest, and Clifford Stein.
\newblock \emph{Introduction to Algorithms}.
\newblock {MIT} Press, 3 edition, 2009.
\newblock ISBN 978-0-262-03384-8.
\newblock URL \url{http://mitpress.mit.edu/books/introduction-algorithms}.

\bibitem[Deligkas et~al.(2021)Deligkas, Eiben, Ganian, Hamm, and
  Ordyniak]{DeligkasEGHO2021}
Argyrios Deligkas, Eduard Eiben, Robert Ganian, Thekla Hamm, and Sebastian
  Ordyniak.
\newblock The parameterized complexity of connected fair division.
\newblock In \emph{Proceedings of the 30th International Joint Conference on
  Artificial Intelligence, {IJCAI}~'21}, pages 139--145. ijcai.org, 2021.
\newblock \doi{10.24963/ijcai.2021/20}.

\bibitem[Deligkas et~al.(2024{\natexlab{a}})Deligkas, Eiben, and
  Goldsmith]{DeligkasEG2024}
Argyrios Deligkas, Eduard Eiben, and Tiger{-}Lily Goldsmith.
\newblock The parameterized complexity of welfare guarantees in {S}chelling
  segregation.
\newblock \emph{Theoretical Computer Science}, 1017:\penalty0 114783,
  2024{\natexlab{a}}.
\newblock \doi{10.1016/j.tcs.2024.114783}.

\bibitem[Deligkas et~al.(2024{\natexlab{b}})Deligkas, Eiben, Knop, and
  Schierreich]{DeligkasEKS2024b}
Argyrios Deligkas, Eduard Eiben, Du{\v{s}}an Knop, and {\v{S}}imon Schierreich.
\newblock Individual rationality in topological distance games is surprisingly
  hard.
\newblock In \emph{Proceedings of the 33rd International Joint Conference on
  Artificial Intelligence, {IJCAI}~'24}, pages 2782--2790. ijcai.org,
  2024{\natexlab{b}}.
\newblock \doi{10.24963/ijcai.2024/308}.

\bibitem[Deligkas et~al.(2024{\natexlab{c}})Deligkas, Eiben, Korchemna, and
  Schierreich]{DeligkasEKS2024}
Argyrios Deligkas, Eduard Eiben, Viktoriia Korchemna, and {\v{S}}imon
  Schierreich.
\newblock The complexity of fair division of indivisible items with
  externalities.
\newblock In \emph{Proceedings of the 38th {AAAI} Conference on Artificial
  Intelligence, {AAAI}~'24}, pages 9653--9661. {AAAI} Press,
  2024{\natexlab{c}}.
\newblock \doi{10.1609/aaai.v38i9.28822}.

\bibitem[Deligkas et~al.(2025)Deligkas, Eiben, Ioannidis, Knop, and
  Schierreich]{DeligkasEIKS2025}
Argyrios Deligkas, Eduard Eiben, Stavros~D. Ioannidis, Dušan Knop, and
  {\v{S}}imon Schierreich.
\newblock Balanced and fair partitioning of friends.
\newblock In \emph{Proceedings of the 39th AAAI Conference on Artificial
  Intelligence, AAAI~'25}, pages 13754--13762. {AAAI} Press, 2025.
\newblock \doi{10.1609/aaai.v39i13.33503}.

\bibitem[Deligkas et~al.(2026)Deligkas, Eiben, Goldsmith, Knop, and
  Schierreich]{DeligkasEGKS2026}
Argyrios Deligkas, Eduard Eiben, Tiger-Lily Goldsmith, Dušan Knop, and
  {\v{S}}imon Schierreich.
\newblock Stability in distance preservation games on graphs.
\newblock In \emph{Proceedings of the 25th International Conference on
  Autonomous Agents and Multiagent Systems, {AAMAS}~'26}. IFAAMAS, 2026.

\bibitem[Erd\H{o}s and Szekeres(1935)]{erdos1935combinatorial}
Paul Erd\H{o}s and George Szekeres.
\newblock A combinatorial problem in geometry.
\newblock \emph{Compositio Mathematica}, 2:\penalty0 463--470, 1935.

\bibitem[Fioravantes et~al.(2025)Fioravantes, Gahlawat, and
  Melissinos]{FioravantesGM2025}
Foivos Fioravantes, Harmender Gahlawat, and Nikolaos Melissinos.
\newblock Exact algorithms and lower bounds for forming coalitions of
  constrained maximum size.
\newblock In \emph{Proceedings of the 39th AAAI Conference on Artificial
  Intelligence, AAAI~'25}, pages 13847--13855. {AAAI} Press, 2025.
\newblock \doi{10.1609/aaai.v39i13.33514}.

\bibitem[Foley(1967)]{Foley1967}
Duncan~Karl Foley.
\newblock Resource allocation and the public sector.
\newblock \emph{Yale Economic Essays}, 7:\penalty0 45--98, 1967.

\bibitem[Friedrich et~al.(2023)Friedrich, Lenzner, Molitor, and
  Seifert]{FriedrichLMS2023}
Tobias Friedrich, Pascal Lenzner, Louise Molitor, and Lars Seifert.
\newblock Single-peaked jump {S}chelling games.
\newblock In \emph{Proceedings of the 16th International Symposium on
  Algorithmic Game Theory, {SAGT}~'23}, volume 14238 of \emph{LNCS}, pages
  111--126. Springer, 2023.
\newblock \doi{10.1007/978-3-031-43254-5\_7}.

\bibitem[Gajarsk{\'{y}} et~al.(2013)Gajarsk{\'{y}}, Lampis, and
  Ordyniak]{GajarskyLO2013}
Jakub Gajarsk{\'{y}}, Michael Lampis, and Sebastian Ordyniak.
\newblock Parameterized algorithms for modular-width.
\newblock In \emph{Proceedings of the 8th International Symposium on
  Parameterized and Exact Computation, {IPEC}~'13}, volume 8246 of \emph{LNCS},
  pages 163--176. Springer, 2013.
\newblock \doi{10.1007/978-3-319-03898-8\_15}.

\bibitem[Ganian et~al.(2023)Ganian, Hamm, Knop, Roy, Schierreich, and
  Such{\'{y}}]{GanianHKRSS2023}
Robert Ganian, Thekla Hamm, Dušan Knop, Sanjukta Roy, {\v{S}}imon Schierreich,
  and Ondřej Such{\'{y}}.
\newblock Maximizing social welfare in score-based social distance games.
\newblock In \emph{Proceedings of the 19th Conference on Theoretical Aspects of
  Rationality and Knowledge, {TARK}~'23}, volume 379 of \emph{{EPTCS}}, pages
  272--286, 2023.
\newblock \doi{10.4204/eptcs.379.22}.

\bibitem[Garey and Johnson(1975)]{GareyJ1975}
Michael~R. Garey and David~S. Johnson.
\newblock Complexity results for multiprocessor scheduling under resource
  constraints.
\newblock \emph{{SIAM} Journal on Computing}, 4\penalty0 (4):\penalty0
  397--411, 1975.
\newblock \doi{10.1137/0204035}.
\newblock URL \url{https://doi.org/10.1137/0204035}.

\bibitem[Gaspers and Najeebullah(2019)]{GaspersN2019}
Serge Gaspers and Kamran Najeebullah.
\newblock Optimal surveillance of covert networks by minimizing inverse
  geodesic length.
\newblock In \emph{Proceedings of the 33rd {AAAI} Conference on Artificial
  Intelligence, {AAAI}~'19}, pages 533--540. {AAAI} Press, 2019.
\newblock \doi{10.1609/aaai.v33i01.3301533}.

\bibitem[Grandi et~al.(2023)Grandi, Kanesh, Lisowski, Sridharan, and
  Turrini]{GrandiKLST2023}
Umberto Grandi, Lawqueen Kanesh, Grzegorz Lisowski, Ramanujan Sridharan, and
  Paolo Turrini.
\newblock Identifying and eliminating majority illusion in social networks.
\newblock In \emph{Proceedings of the 37th {AAAI} Conference on Artificial
  Intelligence, {AAAI}~'23}, pages 5062--5069. {AAAI} Press, 2023.
\newblock \doi{10.1609/aaai.v37i4.25634}.

\bibitem[Harris and Narayanaswamy(2024)]{HarrisN2024}
David~G. Harris and N.~S. Narayanaswamy.
\newblock A faster algorithm for vertex cover parameterized by solution size.
\newblock In \emph{Proceedings of the 41st International Symposium on
  Theoretical Aspects of Computer Science, {STACS}~'24}, volume 289 of
  \emph{LIPIcs}, pages 40:1--40:18, 2024.
\newblock \doi{10.4230/LIPIcs.STACS.2024.40}.

\bibitem[{H{\'{e}}bert{-}Johnson} et~al.(2024){H{\'{e}}bert{-}Johnson},
  Lokshtanov, Sonar, and Surianarayanan]{HerbetHLSS2024}
{\'{U}}rsula {H{\'{e}}bert{-}Johnson}, Daniel Lokshtanov, Chinmay Sonar, and
  Vaishali Surianarayanan.
\newblock Parameterized complexity of kidney exchange revisited.
\newblock In \emph{Proceedings of the 33rd International Joint Conference on
  Artificial Intelligence, {IJCAI}~'24}, pages 76--84. ijcai.org, 2024.
\newblock \doi{10.24963/ijcai.2024/9}.

\bibitem[Hoffmann et~al.(2017)Hoffmann, McCreesh, and Reilly]{HoffmannMR2017}
Ruth Hoffmann, Ciaran McCreesh, and Craig Reilly.
\newblock Between subgraph isomorphism and maximum common subgraph.
\newblock In \emph{Proceedings of the 31st {AAAI} Conference on Artificial
  Intelligence, {AAAI}~'17}, pages 3907--3914. {AAAI} Press, 2017.
\newblock \doi{10.1609/aaai.v31i1.11137}.

\bibitem[Kaklamanis et~al.(2018)Kaklamanis, Kanellopoulos, and
  Patouchas]{KaklamanisKP2018}
Christos Kaklamanis, Panagiotis Kanellopoulos, and Dimitris Patouchas.
\newblock On the price of stability of social distance games.
\newblock In \emph{Proceedings of the 11th International Symposium on
  Algorithmic Game Theory, {SAGT}~'18}, volume 11059 of \emph{LNCS}, pages
  125--136. Springer, 2018.
\newblock \doi{10.1007/978-3-319-99660-8\_12}.

\bibitem[Kanellopoulos et~al.(2023)Kanellopoulos, Kyropoulou, and
  Voudouris]{KanellopoulosKV2022}
Panagiotis Kanellopoulos, Maria Kyropoulou, and Alexandros~A. Voudouris.
\newblock Not all strangers are the same: {T}he impact of tolerance in
  {S}chelling games.
\newblock \emph{Theoretical Computer Science}, 971:\penalty0 114065, 2023.
\newblock \doi{10.1016/j.tcs.2023.114065}.

\bibitem[Knop et~al.(2022)Knop, Schierreich, and Such{\'{y}}]{KnopSS2022}
Du{\v{s}}an Knop, {\v{S}}imon Schierreich, and Ond{\v{r}}ej Such{\'{y}}.
\newblock Balancing the spread of two opinions in sparse social networks
  (student abstract).
\newblock In \emph{Proceedings of the 36th {AAAI} Conference on Artificial
  Intelligence, {AAAI}~'22}, pages 12987--12988. {AAAI} Press, 2022.
\newblock \doi{10.1609/aaai.v36i11.21630}.

\bibitem[Knop and Schierreich(2023)]{KnopS2023}
Dušan Knop and {\v{S}}imon Schierreich.
\newblock Host community respecting refugee housing.
\newblock In \emph{Proceedings of the 22nd International Conference on
  Autonomous Agents and Multiagent Systems, {AAMAS}~'23}, pages 966--975.
  {IFAAMAS}, 2023.
\newblock URL \url{https://dl.acm.org/doi/10.5555/3545946.3598736}.

\bibitem[Korhonen and Parviainen(2015)]{KorhonenP2015}
Janne~H. Korhonen and Pekka Parviainen.
\newblock Tractable bayesian network structure learning with bounded vertex
  cover number.
\newblock In \emph{Proceedings of the 29th Annual Conference on Neural
  Information Processing Systems, {NIPS}~'15}, pages 622--630, 2015.
\newblock URL
  \url{https://proceedings.neurips.cc/paper/2015/hash/66368270ffd51418ec58bd793f2d9b1b-Abstract.html}.

\bibitem[Kreisel et~al.(2024)Kreisel, Boehmer, Froese, and
  Niedermeier]{KreiselBFN2024}
Luca Kreisel, Niclas Boehmer, Vincent Froese, and Rolf Niedermeier.
\newblock Equilibria in {S}chelling games: {C}omputational hardness and
  robustness.
\newblock \emph{Autonomous Agents and Multi-Agent Systems}, 38\penalty0
  (1):\penalty0 9, 2024.
\newblock \doi{10.1007/s10458-023-09632-7}.

\bibitem[Lampis(2012)]{Lampis2012}
Michael Lampis.
\newblock Algorithmic meta-theorems for restrictions of treewidth.
\newblock \emph{Algorithmica}, 64\penalty0 (1):\penalty0 19--37, 2012.
\newblock \doi{10.1007/s00453-011-9554-x}.

\bibitem[Lisowski and Schierreich(2025)]{LisowskiS2025}
Grzegorz Lisowski and {\v{S}}imon Schierreich.
\newblock Stability in newcomers' housing: {A} story about anonymous
  preferences and beyond.
\newblock In \emph{Proceedings of the 22nd European Conference on Multi-Agent
  Systems, {EUMAS}~'25}, LNCS. {Springer}, 2025.

\bibitem[McCreesh et~al.(2018)McCreesh, Prosser, Solnon, and
  Trimble]{McCreeshPST2018}
Ciaran McCreesh, Patrick Prosser, Christine Solnon, and James Trimble.
\newblock When subgraph isomorphism is really hard, and why this matters for
  graph databases.
\newblock \emph{Journal of Artificial Intelligence Research}, 61:\penalty0
  723--759, 2018.
\newblock \doi{10.1613/jair.5768}.

\bibitem[Narayanan et~al.(2025{\natexlab{a}})Narayanan, Opatrny, Tummala, and
  Voudouris]{NarayananOTV2025}
Lata Narayanan, Jaroslav Opatrny, Shanmukha Tummala, and Alexandros~A.
  Voudouris.
\newblock Variety-seeking jump games on graphs.
\newblock In \emph{Proceedings of the 34th International Joint Conference on
  Artificial Intelligence, {IJCAI}~'25}, pages 4005--4013. ijcai.org,
  2025{\natexlab{a}}.
\newblock \doi{10.24963/ijcai.2025/446}.

\bibitem[Narayanan et~al.(2025{\natexlab{b}})Narayanan, Sabbagh, and
  Voudouris]{NarayananSV2025}
Lata Narayanan, Yasaman Sabbagh, and Alexandros~A. Voudouris.
\newblock Diversity-seeking jump games in networks.
\newblock \emph{Autonomous Agents and Multi-Agent Systems}, 39\penalty0
  (2):\penalty0 32, 2025{\natexlab{b}}.
\newblock \doi{10.1007/s10458-025-09714-8}.

\bibitem[Schierreich(2023)]{Schierreich2023}
{\v{S}}imon Schierreich.
\newblock Anonymous refugee housing with upper-bounds.
\newblock \emph{CoRR}, abs/2308.09501, 2023.
\newblock \doi{10.48550/ARXIV.2308.09501}.
\newblock URL \url{https://doi.org/10.48550/arXiv.2308.09501}.

\bibitem[Schierreich(2024)]{Schierreich2024}
{\v{S}}imon Schierreich.
\newblock Two-stage refugee resettlement models: {C}omputational aspects of the
  second stage.
\newblock In \emph{Proceedings of the 7th {AAAI/ACM} Conference on AI, Ethics,
  and Society, {AIES}~'24}, pages 50--51. {AAAI} Press, 2024.
\newblock \doi{10.1609/aies.v7i2.31908}.

\bibitem[Velez(2016)]{Velez2016}
Rodrigo~A. Velez.
\newblock Fairness and externalities.
\newblock \emph{Theoretical Economics}, 11:\penalty0 381--410, 2016.
\newblock \doi{10.3982/TE1651}.

\bibitem[Walsh(1999)]{Walsh1999}
Toby Walsh.
\newblock Search in a small world.
\newblock In \emph{Proceedings of the 16th International Joint Conference on
  Artificial Intelligence, {IJCAI}~'99}, pages 1172--1177. Morgan Kaufmann,
  1999.

\bibitem[Watts and Strogatz(1998)]{WattsS1998}
Duncan~J. Watts and Steven~H. Strogatz.
\newblock Collective dynamics of `small-world' networks.
\newblock \emph{Nature}, 393\penalty0 (6684):\penalty0 440--442, 1998.
\newblock \doi{10.1038/30918}.

\bibitem[Wilczynski(2023)]{Wilczynski2023}
Ana{\"{e}}lle Wilczynski.
\newblock Ordinal hedonic seat arrangement under restricted preference domains:
  {S}wap stability and popularity.
\newblock In \emph{Proceedings of the 32nd International Joint Conference on
  Artificial Intelligence, {IJCAI}~'23}, pages 2906--2914. ijcai.org, 2023.
\newblock \doi{10.24963/ijcai.2023/324}.

\end{thebibliography}

\clearpage

\appendix
\addtocontents{toc}{\protect\setcounter{tocdepth}{1}}

\end{document}